%% file: root.tex
\documentclass[journal,twoside,web]{ieeecolor-modified}

\usepackage{generic}
\usepackage{cite}
\usepackage{graphicx}
\usepackage{amsmath,mathtools}
\usepackage{amssymb,amsfonts}
\usepackage{mathrsfs}
\usepackage{epstopdf}
\usepackage{color}
\usepackage{hyperref}
\usepackage{float}
\usepackage{braket}
\usepackage{algorithmicx}
\usepackage{xifthen}
\usepackage{needspace}
\usepackage{algpseudocode}

\newtheorem{theorem}{Theorem}
\newtheorem{remark}{Remark}

\newtheorem{definition}{Definition}

\newtheorem{assumption}{Assumption}

\newcommand{\mb}{\mathbf}
\newcommand{\mc}{\mathcal}
\newcommand{\mbb}{\mathbb}
\newcommand{\mr}{\mathrm}
\newcommand{\tr}{\textrm}
\newcommand{\mf}{\mathfrak}
\newcommand{\ZN}{\mathbb Z/N\mathbb Z}
\newcommand{\ZQ}{\mathbb Z/Q\mathbb Z}
\newcommand{\ZNm}{(\mathbb Z/N^2\mathbb Z)^\ast}

\renewcommand{\Pr}{\mathrm{Pr}}

\DeclarePairedDelimiter\floor{\lfloor}{\rfloor}
\def\BibTeX{{\rm B\kern-.05em{\sc i\kern-.025em b}\kern-.08em
    T\kern-.1667em\lower.7ex\hbox{E}\kern-.125emX}}

\begin{document}

{\vspace{1cm}
\title{Private Weighted Sum Aggregation}}
\author{Andreea B. Alexandru, \IEEEmembership{Student Member, IEEE} and George J. Pappas, \IEEEmembership{Fellow,~IEEE}
\thanks{The authors are with the Department of Electrical and Systems Engineering, University of Pennsylvania, Philadelphia PA 19105 USA (e-mail: \{aandreea,pappasg\}@seas.upenn.edu).%
}}

\maketitle

\begin{abstract}
As large amounts of data are circulated both from users to a cloud server and between users, there is a critical need for privately aggregating the shared data. This paper considers the problem of \emph{private weighted~sum aggregation} with secret weights, where an aggregator wants to compute the weighted sum of the local data of some agents. Depending on the privacy requirements posed on the weights, there are different secure multi-party computation schemes exploiting the information structure. First, when each agent has a local private value and a local private weight, we review private sum aggregation schemes. 
Second, we discuss how to extend the previous schemes for when the agents have a local private value, but the aggregator holds the corresponding weights. 
Third, we treat a more general case where the agents have their local private values, but the weights are known neither by the agents nor by the aggregator; they are generated by a system operator, who wants to keep them private. We give a solution where aggregator obliviousness is achieved, even under collusion between the participants, and we show how to obtain a more efficient communication and computation strategy for multi-dimensional data, by batching the data into fewer ciphertexts. 
Finally, we implement our schemes and discuss the numerical results and efficiency improvements. 
\end{abstract}

\input{Introduction}
\input{Problem_Statement}
\input{Private_Sum_Aggregation}
\input{Private_Centralized_Weighted_Sum_Aggregation}
\input{Private_Weighted_Sum_Aggregation}
\input{Distributed_Generation}
\input{Packed_Paillier}
\input{Multidimensional_Schemes}
\input{Simulations}
\input{Conclusion}

\bibliographystyle{IEEEtran}
\bibliography{IEEEabrv,biblo}
\addtolength{\textheight}{-7.5cm}   
\input{Appendix}

\end{document}

%% file: Introduction.tex
\section{Introduction}\label{sec:introduction}
\IEEEPARstart{T}{he} recent technological advances in communication speed and deployment of millions of devices have fostered the adoption of distributed computing frameworks. In turn, such frameworks require aggregating services in order to utilize the data collected for specific causes. Even as a step in distributed algorithms, aggregation shifts from a decentralized nature that inherently guarantees more privacy, to a centralized approach that poses severe privacy challenges. 

Of particular interest is the general problem of weighted sum aggregation, that we explore in this paper, in which an aggregating party needs to collect and sum contributions from a number of agents--the contributions consist of some local data weighted by some other relevant quantities. 
There is a wealth of examples spanning various research areas that require the computation of weighted aggregates:
\mbox{\textbf{(a)}}~Decentralized and cooperative linear control for multi-agent systems~\cite{Corfmat1976decentralized,Wang1998robust,Lin2011};
\mbox{\textbf{(b)}}~Graph neural networks~\cite{Zhou2018graph,Kipf2016semi} and collaborative inference~\cite{Teerapittayanon2017distributed,Gupta2018distributed};	
\mbox{\textbf{(c)}}~Average consensus~\cite{Fagnani2009average,Yang2013consensus};	
\mbox{\textbf{(d)}}~Federated learning~\cite{Mcmahan2017federated,Yang2019federated}, aggregation of linear inference results;
\mbox{\textbf{(e)}}~Energy price aggregation and management~\cite{Cleveland2000aggregation,Wang2019incentivizing}, vehicle tolls collection~\cite{Levinson2003model,Balasch2010pretp}.	

Each of the above examples can pose different privacy requirements on the local data of the agents, as well as~on the corresponding weights. For example, in the context of federated learning, the model is locally trained by the agents and the aggregating server needs to compute the mean model without obtaining the local models. In some price collection instances, the prices can depend on private information known at the aggregator and can vary dynamically, so the aggregator knows the price weights, while the agents do not. Finally,~there are cases where a system operator has invested resources~into computing the control gains for a distributed system and~wants to keep them private from both the agents and the aggregator, who needs to compute a linear control without knowing neither the local agent's states nor the gains. Similar privacy requirements are in place for secure inference, where a service provider has trained a proprietary model on its own data and wants to keep it private while allowing it to be deployed.

In the context of \textbf{(a)}, linear distributed control with homomorphically encrypted gains was addressed by~\cite{Darup2019encrypted,Alexandru2019encrypted,Alexandru2020private}, with~\cite{Alexandru2020private} touching also on \textbf{(b)}. We will elaborate and improve on these works in Section~\ref{sec:private_weighted_sum_aggregation}. Concerning \textbf{(c)}, there is a body of research that targets the privacy of the local data of the agents achieving consensus, using partially homomorphic encryption or differential privacy: see~\cite{Ruan2019secure,Hadjicostis2018privacy,Hadjicostis2020privacy,Nozari2017differentially} and the references within. For \textbf{(d)} and \textbf{(e)}, works such as~\cite{Li2010secure,Acs2011have,Kursawe2011privacy,Erkin2013privacy,Bonawitz2017practical} provide private solutions for private sum aggregation, touching on a large base of cryptographic tools, such as secret sharing, threshold homomorphic encryption, differential privacy. 

Given the wide spread of private weighted sum aggregation problems with different privacy constraints, our first~contribution is to review their solutions and give a unified~formulation. We intend for this paper to serve as~a~guide~for~choosing an efficient particular solution based on knowledge distribution and privacy demands. 
Our second contribution is~to~offer~a~private solution for the general case of weighted~aggregation,~where weights are hidden from all parties, and propose three optimi-zations. 
Our third contribution is to implement~and~extensively demonstrate the runtime and communication improvements. 

We first review existing solutions for \emph{private sum aggregation} (when weights are known by the agents, but not by the aggregator). Most of the previously mentioned literature falls into this category. Second, we describe the \emph{private weighted sum aggregation with centralized weights} (when weights are known at the aggregator, but not at the agents), which can be solved from the lens of functional encryption for inner products. 
Third, we give a solution for the \emph{private~weighted~sum aggregation with hidden weights} (neither agents nor aggregator know the weights) and improve it compared to the solution proposed in~\cite{Alexandru2019encrypted} in terms of security: larger collusion threshold, communication: fewer messages exchanged, and runtime: fewer operations, in the case of multi-dimensional data. We also propose multi-dimensional extensions for the first two schemes.

\subsubsection*{Notation}
We use bold-face lower case for vectors, e.g.,~$\mb x$, and bold-face upper case for matrices, e.g., $\mb A$. For a positive~integer~$n$, let 
$[n]:=\{1,2,\ldots,n\}$. A quantity $(\cdot)_i$ refers to agent~$i$~and~a quantity $(\cdot)_a$ refers to the aggregator. By~$\mb x^{[j]}$, we refer to the \mbox{$j$-th} element of vector $\mb x$ and by $\mb W^{[jl]}$, to the element of matrix $\mb W$ on the $j$-th row and $l$-th column. 
$\mbb Z$~denotes the set of integers, $\mbb Z/N\mbb Z$ denotes the additive group of integers modulo $N$ and $(\mbb Z/N\mbb Z)^\ast$ denotes the multiplicative group of integers modulo $N$. 
$\kappa$ is the security parameter. We denote the Paillier encryption primitive by $\mr E(\cdot)$ and the 
decryption primitive by $\mr D(\cdot)$. A function $\eta:\mbb Z_{\geq 1}\rightarrow \mbb R$ is negligible if $\forall c\in \mbb R_{> 0}, \exists n_c\in\mbb Z_{\geq 1}$ such that for all integers $n\geq n_c$, we have $|\eta(n)|\leq n^{-c}$. $\phi(N)$ denotes Euler's totient function; for $N=pq$, with $p,q$ primes, $\phi(N)=(p-1)(q-1)$.  
A value $x\in\mbb Q_{l_i,l_f}$ represents a rational value $x = x_i.x_f$ with $l_i$ bits for the integer part and $l_f$ bits for the fractional part.


%% file: Problem_Statement.tex
\section{Problem statement}\label{sec:problem_statement}
We investigate an aggregation problem of weighted contributions, depicted schematically in Figure~\ref{fig:diagram}. We consider a system with $M$ agents and one aggregator. Each agent $i\in[M]$ has some data $\mb x_i(t)\in\mbb R^{n_i}$ at time $t$ and the aggregator wants to compute an aggregate of the data in the system $\mb x_a(t)\in\mbb R^{n_a}$, where $\mb W_i\in\mbb R^{n_a\times n_i}$ are constant weights designated for the local data of agent~$i$:
\begin{equation}\label{eq:aggregation}
	\mb x_a(t) =  \sum_{i=1}^M \mb W_i \mb x_i(t).
\end{equation} 

At every time step, each agent $i\in [M]$ has access to its local data $\mb x_i(t)$, either by direct measurement (e.g., location, energy consumption) or by computation (e.g., gradient of the model, local prediction). 
We consider three types of privacy requirements for private weighted sum aggregation.
\begin{figure}[b]
\begin{center}
    \vspace{-10pt}
    \includegraphics[width=0.43\textwidth]{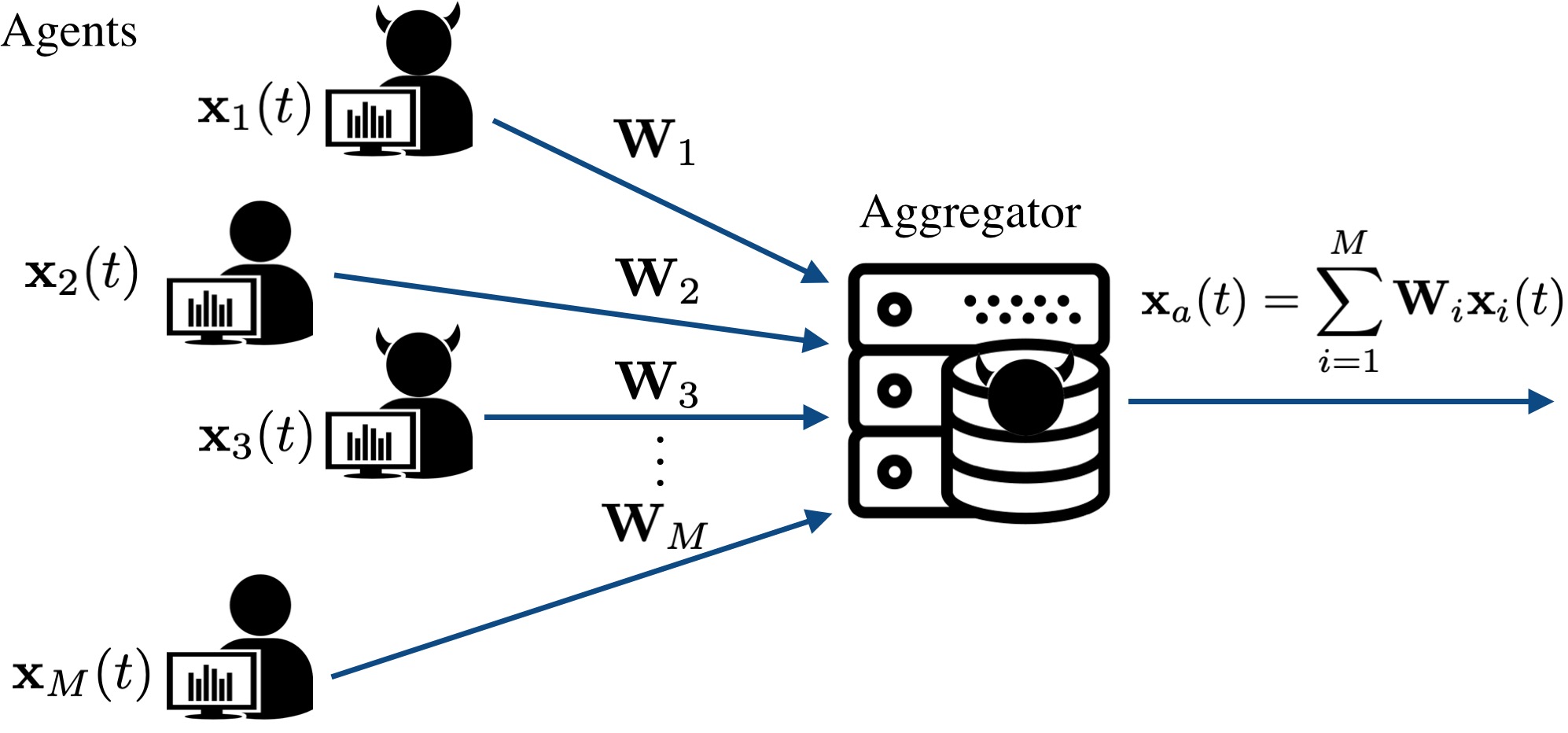}
  \caption{Diagram of the private weighted sum aggregation. Some of the participants can be corrupted and disclose their private data.}
  \vspace{-4pt}
   \label{fig:diagram}
 \end{center}
\end{figure}

\subsubsection*{Private Weighted Sum Aggregation with hidden weights} 
This case requires the strongest privacy guarantees:
\begin{enumerate}
    \item[(a)] Agent~$i$ should not infer anything about the other agents' local data $\mb x_j(t)$, $j\in[M]\setminus \{i\}$ or about the aggregator's result $\mb x_a(t)$ or about the weights $\mb W_i$, $i\in[M]$, including partial information such as $\mb W_i\mb x_i(t)$.
    \item[(b)] The aggregator should only be able to compute $\mb x_a(t)$ and should not infer anything else about the agents' local data $\mb x_i(t)$ or the weights $\mb W_i$, $i\in[M]$, including partial information such as $\mb W_i\mb x_i(t)$.    
\end{enumerate}
\subsubsection*{Private Sum Aggregation}
In this case, we replace (a) by:
\begin{enumerate}
    \item[(a)] Agent~$i$ knows its corresponding weight $\mb W_i$ and should not infer anything about the other agents' local data and weights $\mb x_j(t),\mb W_j,j\in[M]\setminus \{i\}$, including partial information such as $\mb W_j\mb x_j(t)$, or about the aggregator's result $\mb x_a(t)$.
\end{enumerate}
\subsubsection*{Private Weighted Sum Aggregation with centralized weights}
In this case, we replace (b) by:
\begin{enumerate}
    \item[(b)] The aggregator knows the weights $\mb W_i$ and should only be able to compute $\mb x_a(t)$ and should not infer anything else about the agents' local data $\mb x_i(t)$, $i\in[M]$, including partial information such as $\mb W_i\mb x_i(t)$.    
\end{enumerate}

These privacy requirements should hold even under collusion between the aggregator and at most $M-2$ agents, or between $M-1$ agents, i.e., a coalition should not be able to infer the private data of the remaining honest participants. 

We consider computationally bounded adversaries that are \textit{honest-but-curious}, which means that an adversary wants to learn the private data of the honest agents, without diverging from the established protocol. Such a model is chosen because the aggregator is interested in obtaining the correct result of the computation, and for instance, in applications involving pricing, the agents would be fined in they cheat. 

The goal here is to protect the privacy of the inputs~and intermediary computations, but reveal the output to the aggregator. As a side note, all the presented algorithms can support differential privacy, in case the output should also be protected.

In describing the schemes in Sections~\ref{sec:private_sum_aggregation},~\ref{sec:known_weighted_sum_aggregation} and~\ref{sec:private_weighted_sum_aggregation}, we focus on scalar data $w_i,x_i,x_a\in\mbb Z_{\geq 0}$, for $i\in[M]$. After illustrating the functionalities, we provide methods for dealing with multi-dimensional rational data in Sections~\ref{sec:packed_Paillier}, \ref{sec:multidim}. 

%% file: Private_Sum_Aggregation.tex
\section{Private sum aggregation}\label{sec:private_sum_aggregation}
Private sum aggregation ($\mr{pSA}$), introduced 
in~\cite{Shi2011privacy, Rastogi2010differentially}, enables an untrusted aggregator to compute the sum of the~private data contributed by agents, without learning their individual contributions. Improvements in terms of efficiency and functionality of $\mr{pSA}$ have been proposed in~\cite{Chan2012privacy,Joye2013scalable,Benhamouda2016new,Becker2018revisiting,Bonawitz2017practical,Tjell2020private} and the references within. 
The formal definition of \emph{aggregator obliviousness} that $\mr{pSA}$ schemes have to satisfy (informally described in Section~\ref{sec:problem_statement}) was introduced in~\cite{Shi2011privacy} and 
is given as a cryptographic game between an adversary and a challenger, similar to the game we describe in Appendix~\ref{app:privacy_def}. 

When the weights $w_i$ are known to the agents, equation~\eqref{eq:aggregation} can be computed privately with a $\mr{pSA}$ scheme. 
Specifically, in every time step, denoted by $t\in\mbb Z_{\geq 0}$, each agent $i\in[M]$ holds a private value $x_i(t)$ and $w_{i}$. Define $v_{i}(t):= w_{i}x_i(t)$. The aggregator wants to compute the aggregate statistics over the private values: 
$x_a(t) =\sum_{i\in[M]}v_{i}(t)$. 

Let $l$ denote the maximum number of bits of $x_i(t),w_i$, $\forall i\in[M]$. An assumption we make for the rest of the paper is:
\begin{assumption}\label{assum:no_overflow}
For each time step~$t$, after discretization, $x_i(t), w_i, w_ix_i(t),x_a(t)< N$, $\forall i\in[M]$, i.e., there is no overflow for $N$ specified in each scheme.\hfill $\diamond$
\end{assumption}

The most intuitive $\mr{pSA}$ scheme involves secret sharing. Each participant will be given by a trusted dealer at the onset of scheme a secret share of zero for each time step. Each agent will use this share to mask its local data, like a one-time pad--see Appendix~\ref{app:SS}. The aggregator will then sum all the contributions it receives and obtain the desired sum, as the shares of zero will cancel out. The idea of using shares of zero to additively mask  private values in aggregation problems was explored, e.g., in~\cite{Castelluccia2009efficient,Li2014efficient,Bonawitz2017practical,Tjell2020private}.

A private sum aggregation scheme should 
consist of the following algorithms, $\mr{pSA_1} = (\mr{Setup}, \mr{Enc}, \mr{AggrDec})$:
\begin{itemize}
\item $\mr{Setup}(1^\kappa, M, w_{i\in[M]},T)$: take as input the security parameter $\kappa$, the number of agents $M$ and time period $T$, and output public parameters $\mr{prm}$, secret information for each agent $\mr{sk}_i$, $i\in[M]$ and for the aggregator $\mr{sk}_a$. 
This happens as follows: let $N =\max(\kappa,2l + M)$, then, for each time step $t\in[T]$, generate $M+1$ shares of zero: 
\[\sum_{i\in[M]\cup\{a\}} s_i(t) = 0 ,\quad s_i(t)\in \ZN.
\]
Set $\mr{prm} = (\kappa,M)$, $\mr{sk}_i = (s_i(t), w_{i})$, $\mr{sk}_a = (s_a(t))$.

\item $\mr{Enc}(\mr{prm},\mr{sk}_i,t,x_{i}(t))$: take as input the public parameters, agent~$i$'s secret information, the time step and \mbox{the local}~private value. 
Set $v_{i}(t) = w_{i}x_i(t)$ and 
compute the ciphertext:
\[c_i(t) =  v_{i}(t) + s_i(t) \in \ZN.\]

\item $\mr{AggrDec}(\mr{prm},\mr{sk}_a,t,\{c_i(t)\}_{i \in[M]})$: take as input the public parameters, the aggregator's secret information, the time 
step and the ciphertexts of the agents for that time step. The aggregator obtains $x_a(t)=\sum_{i\in[M]}v_{i}(t)$, as follows: 
\begin{align*}
x_a(t) = s_a(t) + \sum_{i\in[M]} c_i(t) \bmod N.
\end{align*}
\end{itemize}

The correctness of $\mr{pSA_1}$ is based on the correct generation of the random shares of zero in $\mr{Setup}$, that cancel out after aggregation. \mbox{Aggregator} obliviousness is based on the perfect security of masking the private data in $\mr{Enc}$ by a one-time pad. 

The scheme $\mr{pSA_1}$ requires a different set of shares of zero for every time step (otherwise, partial information such as differences between private contributions at different time steps is leaked), which can involve elaborate communication, as we will see in Section~\ref{sec:distributed}. On the other hand, the $\mr{pSA_2}= (\mr{Setup}, \mr{Enc}, \mr{AggrDec})$ scheme from~\cite{Joye2013scalable}, that we describe 
next, only requires an initial set of shares of zero. This scheme is based on the Paillier cryptosystem~\cite{Paillier99}, see Appendix~\ref{app:AHE}. 
\begin{itemize}
\item $\mr{Setup}(1^\kappa, M,  \{w_{i}\}_{i\in[M]},T)$: generate $p,q$ to be two equal-size primes and set $N = pq$ with $\gcd(\phi(N),N)=1$, $\floor{\log_2 N} =\kappa$. Define a hash function that acts as a random oracle $H:\mbb Z \rightarrow \ZNm$. Generate $M+1$ shares of zero:
\[
s_a := -\sum_{i\in[M]} s_i ,\quad s_i\in \ZNm. 
\]
Set $\mr{prm} = (\kappa, N, H)$, $\mr{sk}_i = (s_i, w_i)$, $\mr{sk}_a = (s_a)$.

\item $\mr{Enc}(\mr{prm},\mr{sk}_i,t,x_i(t))$: set $v_{i}(t) = w_{i}x_i(t)$ and output:
\[c_i(t) = (1+N)^{v_{i}(t)} H(t)^{s_i} \bmod N^2.\]

\item $\mr{AggrDec}(\mr{prm},\mr{sk}_a,t,\{c_i(t)\}_{i\in[M]})$: take as input the public parameters, the aggregator's secret information, the time 
step and the ciphertexts of the agents for that time step. The aggregator obtains $x_a(t)=\sum_{i\in[M]}v_{i}(t)$, as follows: 
\begin{align*}
V(t) &= H(t)^{s_a} \prod_{i\in[M]} c_i(t) \bmod N^2\\
x_a(t) &= (V(t) - 1)/N.
\end{align*}
\end{itemize}

The correctness of this scheme follows from the generation of the secret shares and from~\eqref{eq:GammaN} in Appendix~\ref{app:AHE}. The aggregator obliviousness property is proved 
in~\cite{Joye2013scalable}. Furthermore,~\cite{Joye2013scalable} shows that the security of the scheme is not impacted when the hash function $H$ takes values in $\mbb Z/N^2\mbb Z$, not in $\ZNm$. 

%% file: Private_Centralized_Weighted_Sum_Aggregation.tex
\section{Private weighted sum aggregation with centralized weights}\label{sec:known_weighted_sum_aggregation}
When the aggregator knows the weight corresponding to each of the agents $w_i$ (and they are not identical), but the agents do not know them, we cannot reuse the above private sum aggregation schemes. 
We operate under the assumption that the constant weights are not chosen in an adversarial way and there is an auditor that checks them beforehand. In this way, we ensure that, for instance, the weights are not chosen to single out only one piece of local data. 

There are two lines of work that investigate this problem. First, in~\cite{Bickson2010peer}, the authors propose a distributed scenario for aggregation in a graph of agents using a threshold cryptosystem. This implies that after receiving contributions from the agents, the aggregator would ask for help in decrypting the aggregate value. The second line of work removes the extra communication required for decryption. Functional encryption~\cite{Boneh11functional} is a generalization of homomorphic encryption and allows a party to compute a functionality over the encrypted data of another party and obtain the desired solution without decryption. One of the few current practical implementations is the functionality of inner products, where one party holds one input and the other party holds the other~\cite{Agrawal2016fully}. 
Here, we formulate our problem of private weighted sum aggregation with weights known by the aggregator in terms of functional encryption for inner product: $x_a(t) = \braket{[w_1,\ldots,w_M],[x_1(t),\ldots,x_M(t)]}$. 

The definition of aggregator obliviousness for this case can be written as a cryptographic game formalizing the requirements in Section~\ref{sec:problem_statement}. Stronger privacy definitions, from the perspective of functional encryption, can be found in~\cite{Agrawal2016fully}. 

We modify $\mr{pSA}_2$ to get a private weighted sum with centralized weights scheme $\mr{pWSAc} = (\mr{Setup},\mr{Enc},\mr{AggrDec})$:
\begin{itemize}
\item$\mr{Setup}(1^\kappa, \{w_{i}\}_{i\in[M]},T)$: given the security parameter $\kappa$, generate two equal-size prime numbers $p,q$ and set 
$N = pq$ such that $\floor{\log_2 N} = \kappa$ and $\gcd(\phi(N),N)=1$. 
The public key is $\mf{pk} = (N)$. Sample $M$ values $s_i\in\ZNm$ and~set: 
\begin{equation}\label{eq:weighted_shares}
s_a = -\sum_{i\in[M]} w_{i}s_i.
\end{equation}
Choose a hash function that acts as a random oracle $H:\mbb Z \rightarrow \ZNm$ (see~\cite{Agrawal2016fully}). Finally, set $\mr{prm} = (\kappa,N,H)$, $\mr{sk}_i = (s_i)$ and 
$\mr{sk}_a = (\{w_{i}\}_{i\in[M]}, s_a)$.

\item $\mr{Enc}(\mr{prm},\mr{sk}_i,t,x_i(t))$: For $x_i(t)\in \ZN$, compute:
\begin{align*}
c_i(t) = (1+N)^{x_i(t)}\cdot H(t)^{s_i}\bmod N^2.
\end{align*}

\item$\mr{AggrDec}(\mr{prm},\mr{sk}_a,t,\{c_i(t)\}_{i\in[M]}, \{w_{i}\}_{i\in[M]})$: compute 
\[V(t) = H(t)^{s_a}\cdot\prod_{i\in[M]} c_i(t)^{w_{i}}\bmod N^2\]
\[x_a(t) = (V(t)-1)/N. \]
\end{itemize}

Correctness follows after expanding $V(t)$:
\[V(t) = H(t)^{s_a}\cdot(1+N)^{\sum_{i\in[M]} w_{i}x_i(t) } \cdot H(t)^{\sum_{i\in[M]} w_{i}s_i .}\]
Using $s_a = -\sum_{i\in[M]} w_{i}s_i$, we obtain that $(V(t)-1)/N =  \sum_{i\in[M]} w_{i}x_i(t)=x_a(t)$, as needed. 
Aggregator obliviousness follows from the proof in~\cite{Joye2013scalable}, where the secret of the aggregator is now $s_a = -\sum_{i\in[M]} w_{i}s_i$ instead of 
$-\sum_{i\in[M]} s_i $, and the aggregator raises the ciphertexts of the participants to the respective power $w_{i}$. A different proof can be found in~\cite{Agrawal2016fully}.

Notice that the keys are independent of the time period~$T$. The $\mr{Setup}$ step can be performed as follows by a third-party dealer that does not need to know the weights~of~the aggregator. The dealer generates $M$ random secrets $s_i$ and~sends one to each agent. The aggregator generates the public and secret key of an additively homomorphic encryption scheme, e.g., Paillier~\cite{Paillier99}, encrypts the weights $w_i$, for $i\in[M]$ and sends them to the dealer. Then, the dealer computes~\eqref{eq:weighted_shares} as:
\[\mr E(s_a) = \prod_{i\in[M]} \mr E(w_i)^{-s_i},\]
and sends it to the aggregator, which then simply decrypts $s_a$. 

%% file: Private_Weighted_Sum_Aggregation.tex
\section{Private weighted sum aggregation with hidden weights}\label{sec:private_weighted_sum_aggregation}
A private weighted sum aggregation scheme for weights unknown to all participants is composed of algorithms $\mr{pWSAh} $ $=  (\mr{Setup},\mr{InitW},\mr{Enc}, \mr{AggrDec})$.  
The formal security definition of aggregator obliviousness is given in Definition~\ref{def:weighted_privacy}~in Appendix~\ref{app:privacy_def} as a cryptographic game. 
This game mimics~the real execution of the scheme, but with a more powerful adversary that can choose both the local data of the corrupted participants and the local data of uncorrupted participants. If even in this case, the adversary is not able to break the privacy of the scheme, then the scheme is private also when multiple participants collude and share their private information, but cannot set the private data of the honest participants.

We first note that, in the context of cooperative control, local control laws of the form~\eqref{eq:aggregation} were considered in \cite{Darup2019encrypted,Alexandru2019encrypted,Alexandru2020private}. 
Specifically, \cite{Darup2019encrypted} introduces a private computation and exchange of the ``input portions'' $\mb v_{i}(t) := \mb W_{i} \mb x_i(t)$, that reveal neither the exact local state $\mb x_i$ nor the local controller matrix $\mb W_{i}$ to the aggregator, but can at least reveal the relative rate of decrease/increase of some signals of the agents over multiple time steps. More details about the information leak can be found in~\cite{Alexandru2019encrypted}, where a solution to the $\mr{pWSAh}$ problem that achieves aggregator obliviousness is proposed. The solution in~\cite{Alexandru2019encrypted} required generating fresh secrets at every time step and proposed an online decentralized method that reduced the collusion threshold between participants. 
Finally, \cite{Alexandru2020private} proposed a theoretical solution that addressed the two issues of~\cite{Alexandru2019encrypted}: it reduced the number of generated secrets, kept the collusion threshold at $M-1$ corrupted participants and reduced the number of messages exchanged between the agents and aggregator. This came at the cost of using a lattice-based homomorphic encryption scheme and augmented learning with error ciphertexts, which can be larger than Paillier ciphertexts and might require more expensive operations. 

Our current work extends the results in~\cite{Alexandru2019encrypted} (and avoids the more complex cryptographic tools in~\cite{Alexandru2020private}) in the following way: we provide an online decentralized method of distributing the secret shares of zero among the agents without reducing the collusion threshold and we propose a more compact and efficient private weighted sum aggregation scheme by packing multiple values in one homomorphic Paillier ciphertext. 

In $\mr{pWSAh}$, the weight $w_{i}$ should be private from~all~participants, so one solution is to encrypt it. Then, agent~$i$ has to be able to send an encryption of the masked product $w_{i}x_i(t)$ to the aggregator, and the latter has to be able to compute and decrypt the result. This suggests the outline in Figure~\ref{fig:diagram_pWSAh}:
\begin{itemize}
\item $w_{i}$ should be encrypted with an additively homomorphic encryption that the aggregator knows how to decrypt;
\item the layer of encryption introduced in $\mr{Enc}$ should be compatible with the inner additively homomorphic layer;
\item the aggregator should not be able to decrypt the individual contributions it receives from the agents, despite having the secret key of the 
homomorphic encryption scheme. 
\end{itemize}

\begin{figure}[htbp]
\begin{center}
    \vspace{-3pt}
    \includegraphics[width=0.45\textwidth]{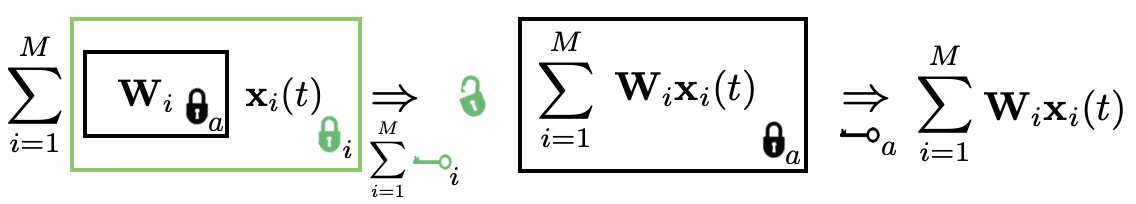}
  \caption{Diagram of the \tr{pWSAh} functionality and privacy requirements.}
  \vspace{-7pt}
   \label{fig:diagram_pWSAh}
 \end{center}
\end{figure}

To achieve the solution, we use a combination of the~two schemes described in Section~\ref{sec:private_sum_aggregation}. For the outer layer~of~encryption, we use one-time pads as in $\mr{pSA_1}$, which are compatible with the additively homomorphic property. For the inner layer of encryption, we use an 
asymmetric additive homomorphic encryption scheme. We instantiate it with the Paillier cryptosystem~\cite{Paillier99}, due to its simplicity and popularity. More details about this cryptosystem can be found in Appendix~\ref{app:AHE}. 

Hence, the steps of the algorithms in $\mr{pWSAh}$ are:
\begin{itemize}
\item $\mr{Setup}(1^\kappa, M, T)$: given the security parameter $\kappa$, get a pair of Paillier keys $(\mf{pk},\mf{sk})$: 
generate two equal-size prime numbers $p,q$ and set $N = pq$ such that $\floor{\log_2 N} = \kappa$ and $\gcd(\phi(N),N)=1$. Set: 
\[\mf{pk} = (N),~\mf{sk} = \big(\phi(N),\phi(N)^{-1}\bmod N\big).\]
For every $t\in[T]$, generate $M+1$ shares of zero: 
\[s_a(t):= -\sum_{i\in[M]} s_i(t),\quad s_i(t)\in (\ZN)^\ast. \]
Finally, set $\mr{prm} = (\kappa,\mf{pk})$, $\mr{sk}_i = (s_i(t\in[T]))$ and $\mr{sk}_a = (\mf{sk}, s_a(t\in[T]))$.

\item $\mr{InitW}(\mr{prm}, M, \{w_{i}\}_{i\in[M]})$: given the public key of the Paillier scheme $pk$, encrypt $w_{i}$ for $i\in[M]$ and return $\mr{sw}_i = \mr{E}(w_{i}) = (1+N)^{w_{i}} r^N \bmod N^2$, for $r$ randomly sampled from $\ZN$ and such that $\gcd(r,N)=1.$

\item $\mr{Enc}(\mr{prm},\mr{sw}_i, \mr{sk}_i,t,x_i(t))$: for $x_i(t)\in \ZN$, compute:
\begin{align*}
c_i(t) = \mr{E}(w_{i})^{x_i(t)} \cdot\mr{E}(s_i(t))
		= \mr{E}(w_{i}x_i(t) + s_i(t)).
\end{align*}

\item $\mr{AggrDec}(\mr{prm},\mr{sk}_a,t,\{c_i(t)\}_{i\in[M]})$: compute $V(t) = \prod_{i\in[M]} c_i(t)\bmod N^2$ and then set: 
\[x_a(t) = \big(\mr{D}(V(t) ) + s_a(t) \big) \bmod N . \]
\end{itemize}

\emph{Correctness}: $\mr{D}(V(t)) = \sum_{i\in[M]} w_{i}x_i(t) + s_i(t) $ follows from the correct execution of Paillier operations in $\mr{Enc}$. 
Then, $\mr{D}(V(t)) + s_a(t) = \sum_{i\in[M]} w_{i}x_i(t) \bmod N = x_a(t)$ from the generation of shares of zero. 

\begin{theorem}\label{thm:scheme_pWSA}
The $\mr{pWSAh}$ scheme achieves weighted sum aggregator obliviousness w.r.t. Definition~\ref{def:weighted_privacy}. \hfill $\diamond$
\end{theorem}

The proof is given in Appendix~\ref{app:proof}.

\begin{remark}\label{rem:shares}
Unlike in $\mr{pSA_2}$, in $\mr{pWSAh}$ the aggregator has to know the secret key of the cryptosystem that encrypts the weights. If we would use $\mr{pSA_2}$, which has a single share of zero per agent for all time steps, an adversary that corrupts the aggregator and selects equal contributions at different time steps for an agent in the $\mr{pWSAO}$ game (described in Appendix~\ref{app:privacy_def}) could learn that agent's secret share.\hfill$\diamond$
\end{remark}

The above scheme is appealing due to its simplicity, but involves demanding communication, because secret shares of zero are required at every time step~$t$ for every participant, as motivated by Remark~\ref{rem:shares}. The $\mr{Setup}$ is executed~by~an~incorruptible trusted third-party, called dealer. This dealer cannot be online at every time step to distribute the shares because, otherwise, this party could act as a trusted aggregator. A more reasonable assumption is that, prior to the online computations, the dealer computes the shares for $T$ time steps and 
sends them to the agents, who have to store them. 
Alternatively, we also offer a solution to generate the secret shares of zero in a distributed way, without the need of a trusted third-party.

%% file: Distributed_Generation.tex
\section{Decentralized generation of zero shares}\label{sec:distributed}
\subsection{One communication round, lower collusion threshold}\label{subsec:one_step}

In~\cite{Alexandru2019encrypted}, we offered a solution with one round of communication but lower collusion threshold. This solution~is~appealing in the case where the agents are connected by a dense graph. Specifically, at each time step, each agent would generate~and send shares to the agents in its neighborhood (including the aggregator), then sum up the shares it received from~its~neighbors. This guarantees that all participants will hold a share of zero. However, if the communication graph between the agents is sparse, the collusion threshold drops from $M-1$ participants to the minimum number of neighbors that an agents has. 

\subsection{Two communication rounds, same collusion threshold}\label{subsec:two_steps}

When agents are not sufficiently connected, there are ways of remediating the problem, each with different trade-offs. 
If we are able to enforce new communication links between the agents, we can create dummy connections such that each agent reaches a desired vertex degree. This keeps the same number of communication rounds as before, but it is debatable whether the cost of adding new communication links is reasonable. For instance, if the connections are based on proximity, such a solution is expensive. 
On the other hand, if we relax the number of communication rounds such that each agent obtains a valid share of zero in two rounds, we can retrieve the initial collusion threshold of $M-1$ participants. Instead of sending the shares to each other, the agents will use the aggregator as an intermediate relay to get to the agents that they are not directly connected to. Specifically, each agent $i\in[M]$ will generate $M+1$ 
shares and encrypt them with a key known the agent $l\in([M]\setminus i)\cup a$ (with, e.g., a symmetric encryption like AES). 
The aggregator will forward the corresponding shares to its neighbors $l\neq i$.  

\begin{enumerate}
	\item At time $t-2$, each agent and the aggregator $i \in [M]\cup a$  creates shares of zero $\sigma_{il}(t)\in (\ZN)^\ast$ for itself and for the rest of the agents:
	\begin{equation}\label{eq:sum_pieces}
		\sum_{l\in [M]\cup a }\sigma_{il}(t) = 0. 
	\end{equation}
	It encrypts them with a key known to agent $l\in [M]\cup a \setminus i$ and sends $\mr{AES}(\mr{key}_l,\sigma_{il}(t))$ to the aggregator. 
	\item At time $t-1$, the aggregator batches the $M$ shares for agent~$i\in[M]$ and sends them.
	\item At time $t$, each agent $l\in [M] \cup a$ sums its own share and the shares it received and decrypted from the aggregator:
	\begin{equation}
		s_{l}(t) := \sum_{i\in [M]\cup a} \sigma_{il}(t) \bmod N.
	\end{equation}
\end{enumerate}

In order to reduce the load on the aggregator, an agent~$i$ can communicate directly to agents $l\in \mc N_i\cap [M]$, where $\mc N_i$ is the set of neighbors of agent~$i$ and only sends the encrypted shares to the aggregator for the agents $l\notin \mc N_i\cap [M]$.

There is a very small probability that $s_i(t) \in \ZN \setminus (\ZN)^\ast$, i.e., it is a multiple of $p$ or $q$. In this case, the aggregator might be able to retrieve the encrypted message with a probability $\leq 1/N$ when the encrypted message spans all $\ZN$ and $\leq 1/\min(p,q)$ when the encrypted message is in a smaller space, e.g. on $l$ bits. However, $p,q,N$ have over a thousand bits to ensure semantic security of the Paillier scheme, so this probability is very small ($\leq$ the probability of brute force guessing the message). Nevertheless, we can introduce an extra round of communication to ensure $s_i(t)\in(\ZN)^\ast$, $i\in[M]$: agent $i$ verifies if $\gcd(s_i(t),N)=1$ and if not, it changes the values $\sigma_{ii}(t)$ and $\sigma_{ai}(t)$ such that~\eqref{eq:sum_pieces} is still satisfied and $\gcd(s_i(t),N)=1$, then sends the new $\sigma_{ai}(t)$ to the aggregator. This works because $s_a(t)$ is not used for masking, so it is not required to be in $(\ZN)^\ast$.

%% file: Packed_Paillier.tex
\section{Packed Paillier scheme}\label{sec:packed_Paillier}
Assume we have a vector $\mb y = [\mb y^{[1]},\mb y^{[2]},\ldots, \mb y^{[m]}]$, with $\mb y^{[i]} \in [0,~2^l) \cap \mbb Z_{\geq 0}$. 
We can take advantage of the fact that $N>>2^l$, where $N$ is the Paillier modulus by packing $m$ items of $l$ bits into a plaintext in $\mbb Z_N$ in the following way:
\begin{equation*}
p_y = \sum_{i=1}^m \mb y^{[i]} 2^{l(i-1)} = [\mb y^{[1]} | \mb y^{[2]} | \ldots | \mb y^{[m]}].
\end{equation*}
Here, we depict the least significant bits on the left, to show the elements in the order that they appear in the vector. If we need to perform additional operations on $p_y$ after packing, we have to make sure we retrieve the correct elements. Hence, we need to take into account possible overflows from one ``slot" of the ciphertext the another. We do this by padding with enough zeroes, where $\delta > l$ and will be determined based on the computations performed on $p_y$ afterwards:
\begin{equation}\label{eq:packing}
p_y = \sum_{i=1}^m \mb y^{[i]} 2^{\delta(i-1)} = [\underbrace{\mb y^{[1]}}_{l} \underbrace{0 ... 0}_{\delta-l}| \underbrace{\mb y^{[2]}}_{l} \underbrace{0 ... 0}_{\delta-l}| \ldots | \underbrace{\mb y^{[m]}}_{l} \underbrace{0 ... 0}_{\delta-l}].
\end{equation}
Note that we can perform~\eqref{eq:packing} as long as $m\delta < N$. 

In~\eqref{eq:packing}, we require positive integers. For a value $y\in \mbb Q_{l_i,l_f}$, we first construct an integer $\overline y$, where $l:=l_i+l_f$, and then obtain a positive integer $\tilde y$, for $\gamma > l$ that we will specify later:
\begin{align}
    \overline y &:= y2^{l_f}  \Rightarrow \bar y \in [-2^{l-1},2^{l-1}) \cap \mbb Z\label{eq:fractional}\\
\tilde y &:= \overline y + 2^{\gamma}  \Rightarrow \tilde y \in [0,2^{\gamma+1}) \cap \mbb Z_{\geq 0}. \label{eq:negative}    
\end{align}

In Sections~\ref{subsec:multidim_PWSAc} and~\ref{subsubsec:naive}, where we do not use~packing, instead of~\eqref{eq:negative} we use (assuming $2^{l-1} < N/2 $): \begin{equation}\label{eq:negativeN}
\tilde y := \begin{cases}\bar y~~&\text{if~~}\bar y \geq 0\\ \bar y + N~~&\text{if~~}\bar y < 0\end{cases} \Rightarrow \tilde y\in\ZN.
\end{equation}

Batching multiple entries into one Paillier ciphertext was first proposed in~\cite{Ge2007answering}. 
Notice that after packing multiple~plaintexts into one Paillier ciphertext as in~\eqref{eq:packing}, we can still~perform the homomorphisms corresponding to element-wise addition and scalar multiplication. 
In the following, we investigate a more efficient way to compute an encrypted matrix-plaintext vector multiplication, by using only packing, element-wise addition and scalar multiplication. Figure~\ref{fig:column} depicts this method. For a matrix $\mb W\in\mbb R^{m\times n}$, denote the $j$th column by $\mb w^j$, for $j\in[n]$. 
Then, in order to obtain the product $ \mb v:=\mb W\mb x $, we multiply each column $\mb w^j$ by the corresponding element in the vector $\mb x^{[j]}$ and 
then sum over all the obtained vectors:
\begin{equation}\label{eq:sumcolumns}
\mb v = \sum_{j=1}^n \mb w^j\mb x^{[j]}.
\end{equation}

\begin{figure}[!ht]
\begin{center}
  \vspace{-8pt}
    \includegraphics[width=0.24\textwidth]{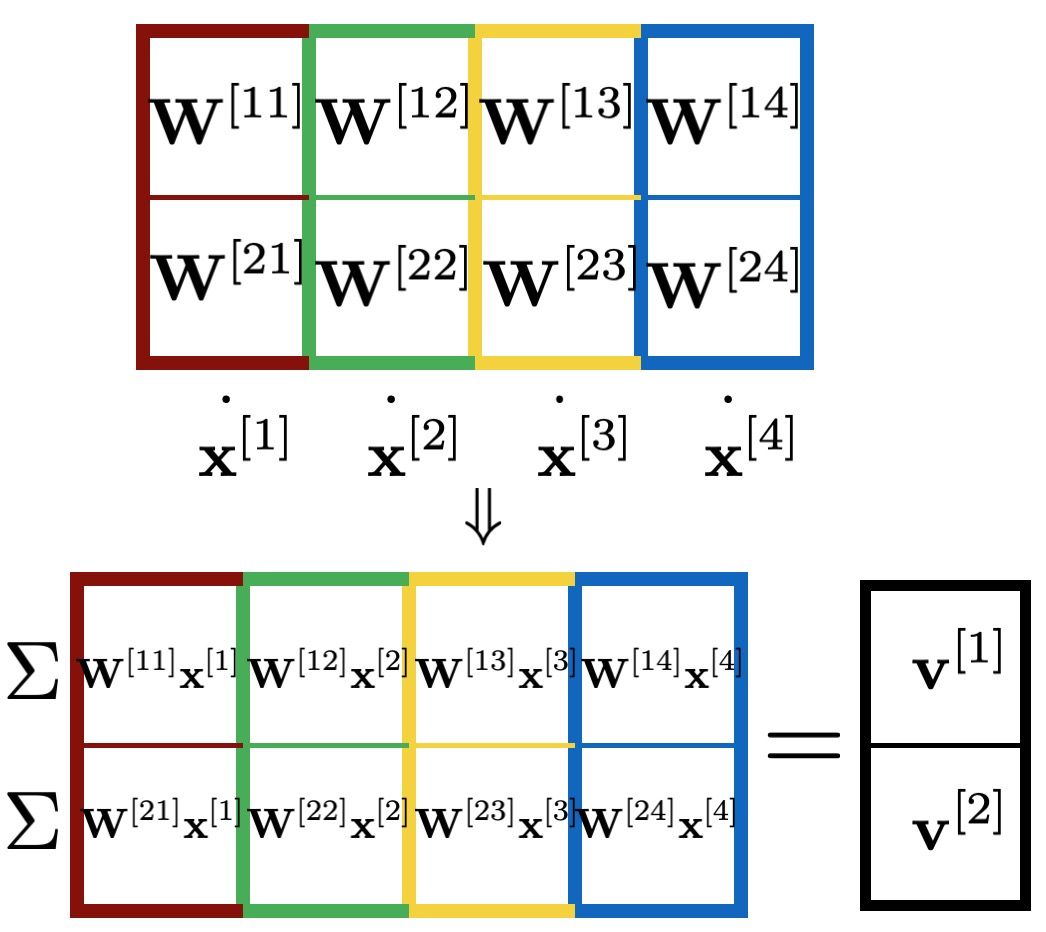}
  \caption{Diagram of column-packed matrix-vector multiplication. The entries with the same outer coloring are packed in the same ciphertext. }
   \vspace{-10pt}
   \label{fig:column}
 \end{center}
\end{figure}

%% file: Multidimensional_Schemes.tex
\section{Schemes for multi-dimensional data}
\label{sec:multidim}
\subsection{Multi-dimensional \tr{pSA}}\label{subsec:pSA_multidim}

In the case of $\mr{pSA}_1$, the messages have small sizes and communication is less of a problem even for multi-dimensional data. 
However, in the case of $\mr{pSA}_2$, messages (ciphertexts) are larger and we propose a better method than sending a different message for each element of the resulting vector, by batching the elements in a single ciphertext. 
The aggregator wants to obtain $\mb x_a(t) = \sum_{i\in[M]} \mb W_i\mb x_i(t) $ $ =: \sum_{i\in[M]}\mb v_i(t)$. In $\mr{pSA}_2$, the dealer generates the secret shares the same way as previously, but each agent $i\in[M]$ computes $\mb v_i^{[j]}(t)$ and uses~\eqref{eq:fractional} and~\eqref{eq:negative} to obtain $\tilde{\mb v}_i^{[j]}(t)$, then computes:
\begin{align*}
    p_i(t) &= \sum_{j\in[n_i]} \tilde{\mb v}_i^{[j]}(t) 2^{\delta(j-1)}\\
    c_i(t) &= (1+N)^{p_i(t)} H(t)^{s_i}\bmod N^2.
\end{align*}
The aggregator computes $V(t)$ as before, and retrieves:
\[\tilde{\mb x}_a^{[k]} = V(t)//2^{(n_a-k)\delta} \bmod 2^{(k-1)\delta},~~k\in\{2,\ldots,n_a-1\}\]
and $\tilde{\mb x}_a^{[1]} = V(t) \bmod 2^{\delta}$, $\tilde{\mb x}_a^{[n_a]} = V(t)//2^{(n_a-1)\delta} $, where by $//$ we mean the quotient operation. From the elements of $\tilde{\mb x}_a(t)$, the aggregator needs to subtract $2^\gamma M$ and divide by $2^{2l_f}$ each element, in order to obtain $\mb x_a(t)$.

Choosing $\gamma = 2l+1$ and $\delta = 2l+2+\lceil\log_2 M\rceil$ ensures the correctness of the decryption, as no overflow occurs.

\subsection{Multi-dimensional \tr{pWSAc}}\label{subsec:multidim_PWSAc}

Here, we cannot use packing to reduce the number of ciphertexts because we would require rotations and element-wise multiplications, which cannot be performed on packed Paillier ciphertexts. 
Compared to the \tr{pWSAh} scheme, \tr{pWSAc} has the advantage that only one set of secret shares are needed for all time steps, hence, communication due to secret generation and sharing only happens once.

In the multi-dimensional case, the algorithms change from the ones in Section~\ref{sec:private_sum_aggregation} as described next. The participants prepare their data using~\eqref{eq:fractional} and~\eqref{eq:negativeN}. In \tr{Setup}, $n_aM$ secrets $\mb s_i\in\left((\ZN)^\ast\right)^{n_a}$ are generated and: 
\[\mb s_a^{[k]} = -\sum_{i\in[M]} \sum_{j\in [n_i]} \mb W_{i}^{[kj]} \mb s_i^{[j]} ,~~k\in[n_a].\]
In \tr{Enc}, each agent $i\in[M]$ constructs $n_a$ ciphertexts:
\[ \mb c_i^{[j]}(t) = (1+\tilde{\mb x}_i^{[j]}(t)N)\cdot H(t)^{\mb s_i^{[j]}}\bmod N^2, ~~j\in[n_i].\]
Finally, in \tr{AggrDec}, the aggregator computes for $k\in[n_a]$:
\[\mb V^{[k]}(t) = H(t)^{\mb s_a^{[k]}}\cdot\prod_{i\in[M]}\prod_{j\in [n_i]} \mb c_i^{[j]}(t)^{\tilde{\mb W}_{i}^{[kj]}}\bmod N^2\]
\[\tilde{\mb x}_a^{[k]}(t) = \left(\mb V^{[k]}(t)-1\right)/N. \]

The aggregator retrieves the elements of $\mb x_a(t)$ from $\tilde{\mb x}_a(t)$ by subtracting~$N$ from the elements greater than $N/2$ and dividing all of them by $2^{2l_f}$.

\subsection{Multi-dimensional \tr{pWSAh}}\label{subsec:multidim_scheme}
We consider values on $l$ bits, with $\log_2 N >> l$, for $N$ ensuring semantic security of the Paillier scheme. Sampling a random value from a large $\ZN$ is expensive, but also redundant, since it masks a much smaller message. To this end, we prefer to sample $s_i(t)\in (0,2^{\lambda+2l}), \forall i\in[M]$, for $\lambda$ the statistical security parameter and to set $s_a(t) := -\sum_{i\in[M]} s_i(t)$ in $\mr{pWSAh}$. 
Masking by $s_i(t)$ will guarantee $\lambda$-statistical security rather than perfect security, see Appendix~\ref{app:SS}. From here on, we use this more efficient approach.

\subsubsection{Naive multi-dimensional scheme}
\label{subsubsec:naive}
This solution was proposed in~\cite{Alexandru2019encrypted}. The algorithms change compared to Section~\ref{sec:private_weighted_sum_aggregation} as follows. In \tr{Setup}, for every $t\in[T]$, $M\cdot n_a$ shares of zero $\mb s_i^{[k]}(t)\in (0,2^{\lambda+2l})$ are generated for $i\in[M], k\in[n_a]$:  
\[\mb s_a^{[k]}(t) = -\sum_{i\in[M]}\sum_{k\in[n_a]} \mb s_i^{[k]}(t).
\]
In $\mr{InitW}$, the weights $\mb W_{i}$, $i\in[M]$ are processed as in~\eqref{eq:fractional} and \eqref{eq:negativeN} and encrypted element-wise: $\mr{E}(\mb W^{[kj]}_{i}) = (1+N)^{\tilde{\mb W}^{[kj]}_{i}}r^N$ $\bmod N^2$, for $r$ randomly sampled from $\ZN$ and satisfying $\gcd(r,N)=1$. In \tr{Enc}, each agent $i\in[M]$ computes:
\begin{align*}
\mb c^{[k]}_i(t) &= \prod_{j\in[n_i]}\mr{E}(\tilde{\mb W}^{[kj]}_{i})^{\tilde{\mb x}^{[j]}_i(t)} \cdot \mr{E}(\mb s^{[k]}_i(t))\\
		&= \mr{E}\left(\sum_{j\in[n_i]}\mb W^{[kj]}_{i}\mb x^{[j]}_i(t) + \mb s^{[k]}_i(t)\right),\forall k\in[n_a].
\end{align*}
Finally, in \tr{AggrDec}, the aggregator computes, for $k\in[n_a]$:
\begin{align*}
\mb V^{[k]}(t) &= \prod_{i\in[M]} \mb c^{[k]}_i(t)\bmod N^2\\
\tilde{\mb x}^{[k]}_a(t) &= \mr{D}(\mb V^{[k]}(t) ) + \mb s^{[k]}_a(t). 
\end{align*}

\subsubsection{Packed multi-dimensional scheme}
\label{subsubsec:packed}
We reduce the number of ciphertexts and the corresponding number of operations by using packing and the more efficient encrypted matrix-plaintext vector multiplication described in Section~\ref{sec:packed_Paillier}. 

Assume at the moment that we can pack at least~$n_a$~values in one ciphertext. The steps we take are:
1) Pre-process the values to be positive and integer; 
2) Pack and encrypt~the~columns of the matrix $\mb W_i$ and obtain $n_i$ ciphertexts; 
3) Perform a scalar multiplication of one encrypted column $c$ with the scalar $\mb x_i^{[c]}(t)$; 
4) Sum the $n_i$ ciphertexts to get the encryption of $\mb W_i\mb x_i(t)$; 
5) Add the share of zero and mask the intermediate results; 
6) Sum the $M$ ciphertexts to obtain the encryption~of $\sum_{i\in[M]}\mb W_i\mb x_i(t)$;
7) Decrypt, unmask and unpack the result.

Next, we detail these steps and depict the bit gains due to the operations performed on the packed columns in Figure~\ref{fig:packed_ops}. 

\begin{figure*}[btp]
\centering
   \includegraphics[width=0.76\textwidth]{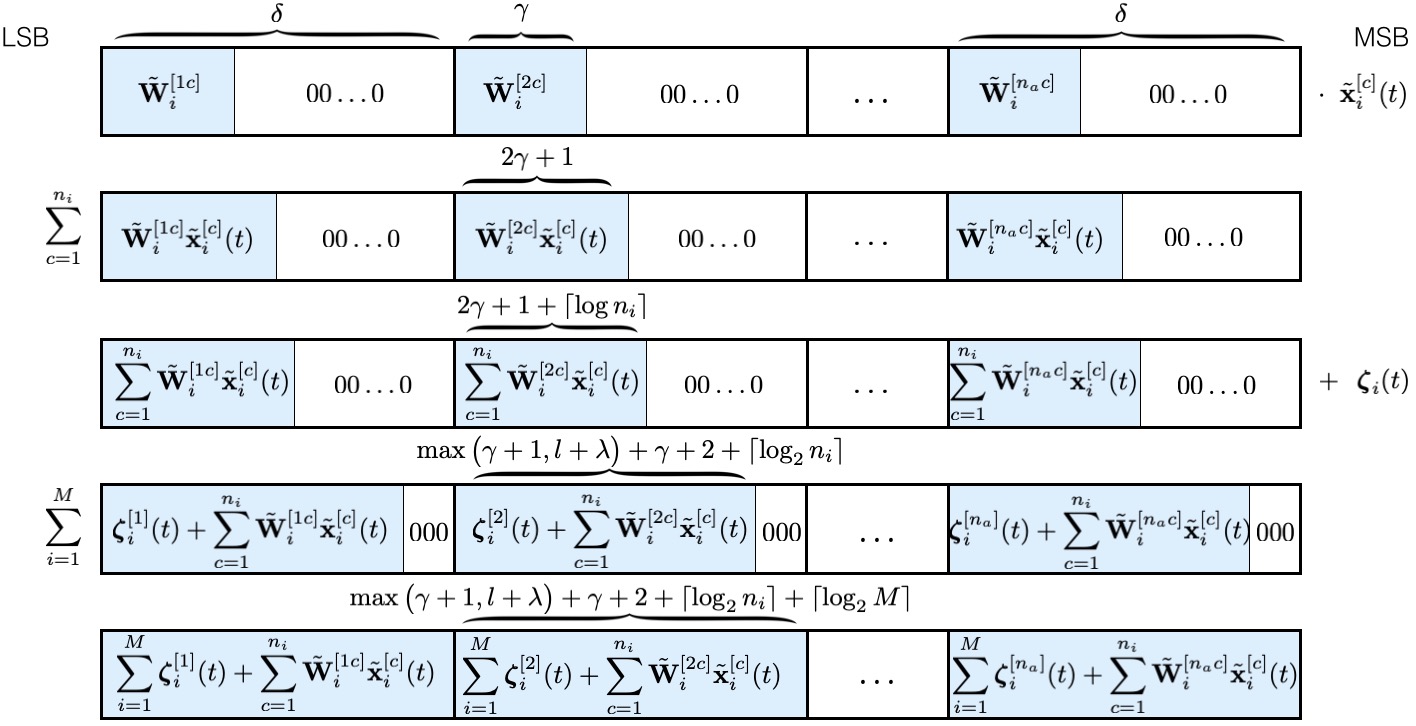}
  \caption{The operations performed on the packed columns and the corresponding number of bits of the result, where on the third line $\boldsymbol{\zeta}_i(t) = \left[\boldsymbol{ \zeta}_i^{[1]}(t) \, 0 \ldots 0 \,\boldsymbol{\zeta}_i^{[2]}(t) \, \ldots \, \boldsymbol{\zeta}_i^{[n_a]}(t) \, 0 \ldots 0 \right]$ and $\boldsymbol \zeta_i^{[k]}(t) = \mathbf s_i^{[k]}(t) + 2^\gamma \mathbf z^{[k]}_i(t)$, for $k\in[n_a]$.}
  \label{fig:packed_ops}
  \vspace{-10pt}
\end{figure*}

1) Prepare the values to be packed as per~\eqref{eq:fractional} and~\eqref{eq:negative}, with~$\gamma$ to be determined later.

2) Construct the packed plaintext $p_{i}^{c}$; encrypt it in $\mr{E}(\mb W_{i}^{c})$:
\[p_{i}^{c} = \sum_{k=1}^{n_a} {\tilde{\mb W}}_{i}^{[kc]} 2^{(k-1)\delta}. \]

3) Multiply the encrypted column $\mr{E}(\mb W_{i}^{c})$ by a pre-processed scalar ${\tilde{\mb x}}_i^{[c]}(t)$. One slot of the ciphertext 
will now contain~\eqref{eq:slotmult} and will be represented on $2\gamma+2$ bits: 
\begin{align}\label{eq:slotmult}
\begin{split}
&({\overline{\mb W}}_{i}^{[kc]} + 2^\gamma)({\overline{\mb x}}_i^{[c]}(t) + 2^\gamma) = \\
&= {\overline{\mb W}}_{i}^{[kc]}{\overline{\mb x}}_i^{[c]}(t) + 2^{2\gamma} + 2^\gamma({\overline{\mb W}}_{i}^{[kc]}+{\overline{\mb x}}_i^{[c]}(t)).
\end{split}
\end{align}
From~\eqref{eq:slotmult}, one can retrieve the desired result ${\overline{\mb W}}_{i}^{[kc]}{\overline{\mb x}}_i^{[c]}(t)$~by taking the lefthand 
side modulo $2^\gamma$. But one can also obtain:
\begin{equation}\label{eq:add_noise}
{\overline{\mb W}}_{i}^{[kc]}+{\overline{\mb x}}_i^{[c]}(t) = \lfloor ({\tilde{\mb W}}_{i}^{[kc]}{\tilde{\mb x}}_i^{[c]}(t) -2^{2\gamma})//2^\gamma \rfloor,
\end{equation}
which can reveal intermediate information to the decryptor. In order to avoid this information leakage, we need to artificially add some noise 
$\mb z^{[kc]}_{i}(t)$ 
that still allows retrieving ${\overline{\mb W}}_{i}^{[kc]}{\overline{\mb x}}_i^{[c]}(t)$. It is more efficient to add this noise in step 5), after performing the sum over $n_i$. 

4) Sum the $n_i$ ciphertexts to output a ciphertext that contains the vector result of the matrix-vector multiplication product. 

5) For each slot $k\in[n_a]$, an agent $i\in[M]$ selects $\mb z^{[k]}_{i}(t) \in (0,2^{l+1+\lambda +\lceil \log_2 n_i \rceil})$, such that a statistical security of $\lambda$ bits is guaranteed for the private value $\sum_{c=1}^{n_i}{\overline{\mb W}}_{i}^{[kc]}+{\overline{\mb x}}_i^{[c]}(t)$. 
Then, each agent constructs its ciphertext $c_i(t)$ by adding the secret shares of zero such that the remaining private value $ \sum_{c=1}^{n_i} {\overline{\mb W}}_{i}^{[kc]}{\overline{\mb x}}_i^{[c]}(t)$ 
is concealed: 
{ \medmuskip=0mu\thickmuskip=1mu
\begin{align}
&c_i(t) := \mb s^{[k]}_{i}(t) + 2^\gamma \mb z^{[k]}_{i}(t) + \sum_{c=1}^{n_i} {\tilde{\mb W}}^{[kc]}_{i}{\tilde{\mb x}}^{[c]}_i(t) \stackrel{\eqref{eq:slotmult}}{=} n_i2^{2\gamma} +\mb s^{[k]}_{i}(t) +\nonumber\\
&+\sum_{c=1}^{n_i} {\overline{\mb W}}_{i}^{[kc]}{\overline{\mb x}}_i^{[c]}(t) +2^\gamma \big( \mb z^{[k]}_{i}(t) +\sum_{c=1}^{n_i}{\overline{\mb W}}_{i}^{[kc]}+{\overline{\mb x}}_i^{[c]}(t) \big).\label{eq:addshare}
\end{align}
}
Due to the masking with $2^\gamma z^{[k]}_{i}(t)$, we can 
reduce the size of the mask $\mb s^{[k]}_{i}(t)$. More specifically, $\mb s^{[k]}_{i}(t)$ can be sampled uniformly at random from $(0,2^\gamma)$, because it acts like a one-time pad (perfect secrecy) on $\sum_{c=1}^{n_i}\overline{\mb W}_{i}^{[kc]}\overline{\mb x}_i^{[c]}(t)$ once the decryptor takes equation~\eqref{eq:addshare} modulo $2^\gamma$. The whole quantity $\mb s^{[k]}_{i}(t) + 2^\gamma \mb z^{[k]}_{i}(t)$ is used for masking, but we only need to ensure that $\sum_{i\in[M]\cup a}\mb s^{[k]}_{i}(t) = 0$. 

6) The aggregator obtains $c(t)$ by taking the product of all ciphertexts $c_i(t)$ it received. 

7) The aggregator decrypts the ciphertext $c(t)$, adds its own secret share of zero $\mb s_a(t)$. It then retrieves the $n_a$ elements of $\tilde{\mb x}_a(t)$ by recursively taking the quotient and rest by $2^\delta$, and from each resulting element, it obtains the elements of $\mb x_a(t)$ by taking modulo $2^\gamma$ and dividing by $2^{2l_f}$.

We now compute the number of bits and corresponding padding one slot can have such that no overflow occurs during the private weighted sum aggregation. We assume that the values of $n_i$, for $i\in[M]$ are similar and define $n:=\max\limits_{i\in[M]}n_i$. 
The value that is packed in slot $k\in[n_a]$ after step 6) is:
\[
\sum_{i=1}^M \left(\mb s^{[k]}_{i}(t) + 2^\gamma \mb z^{[k]}_{i}(t) + \sum_{c=1}^{n_i} \tilde{\mb W}_{i}^{[kc]} \tilde{\mb x}_i^{[c]}(t)\right)  < 2^\delta,
\]
from which we obtain that:
\[\delta > \max\big(\gamma + 1, l+ \lambda \big)+ \lceil \log_2 n \rceil + \lceil \log_2 M \rceil + \gamma + 2.\]
For the correct retrieval of the desired result, we require that:
\begin{align}\label{eq:retrieve_gamma}
\begin{split}
&\sum_{i=1}^M \sum_{c=1}^{n_i} \overline{\mb W}_{i}^{[kc]} \overline{\mb x}_i^{[c]}(t) < 2^\gamma.
\end{split}
\end{align}
From~\eqref{eq:retrieve_gamma}, we choose $\gamma=2l + 1 + \lceil \log_2 n \rceil + \lceil \log_2 M \rceil$ and: 
\begin{align}\label{eq:delta}
\begin{split}
\delta=&\max\big(l+ 2+ \lceil \log_2 n \rceil + \lceil \log_2 M \rceil, \lambda \big) +  
\\ +&3l + 4+ 2(\lceil \log_2 n \rceil + \lceil \log_2 M \rceil).
\end{split}
\end{align}

\begin{figure}[!b]
\centering
\vspace{-10pt}
    \includegraphics[width=0.425\textwidth]{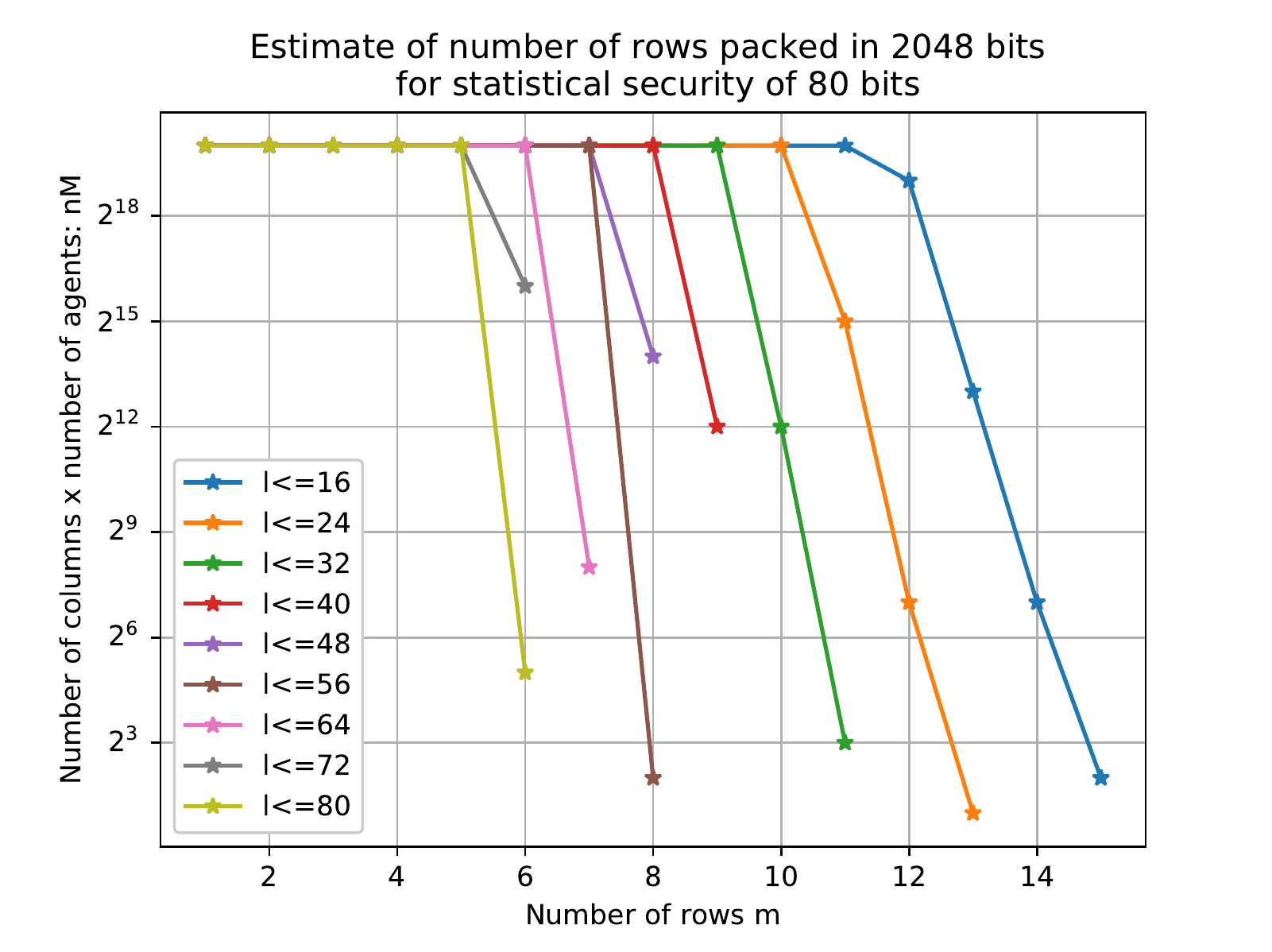}
  \caption{The number of rows $m$ that can be packed in a plaintext of $N=2048$ bits, as a function of the number of columns $n$, number of agents $M$, size of message $l$ and statistical security of $\lambda = 80$ bits.}
   \label{fig:sizes}
\end{figure}

Figure~\ref{fig:sizes} indicates possible values for the number of rows, columns and agents depending on the plaintext size~and~statistical security. Denote the maximum number of values we can pack by $m < N/\delta$. If $n_a > m$, we split the columns into $\lceil n_a/m\rceil$ Paillier ciphertexts and follow the same operations as before, and concatenate the resulting vectors after decryption. 

Let $\mr{pWSAh}^\ast = (\mr{Setup},\mr{InitW}, \mr{Enc}, \mr{AggrDec})$ be a packed multi-dimensional scheme for private weighted sum aggregation with hidden weights, where steps 1) and 2) are performed by the dealer as part of \tr{InitW} and the shares of zero for step 5) are generated as part of \tr{Setup}, steps 1), 3)--5) are performed by each agent $i\in[M]$ in \tr{Enc} and steps 6) and 7) are performed by the aggregator in \tr{AggrDec}. 

\begin{theorem}\label{thm:correctness_multi}
The packed multi-dimensional scheme $\mr{pWSAh}^\ast$ is correct and achieves aggregator obliviousness.
\end{theorem}

\begin{proof}
The correctness follows from the appropriate padding and packing to avoid overflow, as stated in~\eqref{eq:delta}. The aggregator obliviousness proof follows from Theorem~\ref{thm:scheme_pWSA}, along with the intermediate values masking from~\eqref{eq:addshare}.
\end{proof}

\subsection{Comparison between naive and packed method}

In the naive version of the multi-dimensional \tr{pWSAh}, each agent receives $n_an_i$ ciphertexts for $\mb W_{i}$ at the initialization of the protocol, then computes $n_an_i$ ciphertext--scalar multiplications (modular exponentiation), $n_a(n_i-1)$ ciphertext additions (modular multiplications) and sends to the aggregator $n_a$ ciphertexts. In the decentralized way of generating shares, each agent will have to send out $(2l+\lambda)n_a$ bits of randomness to each of the $M$ neighbors per time step.

Denote by $m$ the number of elements we can pack in a Paillier ciphertext. 
In the packed version $\mr{pWSAh}^\ast$, each agent receives $\lceil n_a/m \rceil n_i$ ciphertexts for $\mb W_{i}$ at the initialization of the protocol, then computes $\lceil n_a/m \rceil n_i$ ciphertext--scalar multiplications, $\lceil n_a/m \rceil (n_i-1)$ ciphertext additions and sends to the aggregator $\lceil n_a/m \rceil$ ciphertexts. In the decentralized way of generating shares, each agent will have to send out $(2l +1+\lceil n \rceil + \lceil M \rceil)\lceil n_a/m \rceil$ bits of randomness to each of the $M$ neighbors per time step.

%% file: Simulations.tex
\section{Case study: private distributed control}\label{sec:simulations}
For illustration purposes, we consider a distributed linear control scheme for linear dynamics of $M$ agents in a network:
\begin{equation}\label{eq:dynamics}
	\mb x_i(t+1) = \mb A_{i} \mb x_i(t) + \mb B_i \mb u_i(t), \quad \mb x_i(0) = \mb x_{i,0},
\end{equation}
with $\mb x_i \in \mbb R^{n_i}$ and $\mb u_i \in \mbb R^{m_i}$, for every $i\in[M]$. 
The agents are part of an undirected connected communication graph $\mc G = (\mc V, \mc E)$, with vertex set $\mc V = [M]$ and edge set $\mc E \subseteq \mc V \times  \mc V$. An edge $(i,j) \in \mc E$ specifies that agent~$i$ can communicate with agent~$j$, i.e., agent~$i$ and agent~$j$ are neighbors. 

We can use the local control laws to stabilize the systems: 
\begin{equation}\label{eq:control}
	\mb u_i(t) = \mb K_{ii} \mb x_i(t) + \sum_{j\in\mc N_i} \mb K_{ij} \mb x_j(t),
\end{equation}
where $\mc N_i := \{j\in \mc V | (i,j) \in \mc E\}$ represents the set of neighbors of agent~$i$. 
The stabilizing local control feedback gains $\mb K_{ij}$ can be designed to take into account the structural constraints of the communication graph, see, e.g.~\cite{Lin2011}. 

In this illustrative example, each agent~$i$ is the aggregator of the contributions of its neighbors $\mb K_{ij} \mb x_j(t)$. There is a system operator that acts as the dealer, who designs and encrypts the control feedback weights $\mb K$. 
Specifically, consider a network of 50 agents, with each agent having local states and local control inputs both of dimension 6. We simulate \tr{pWSAh}$^\ast$ for various values of the average node degree in the network, obtained by varying the probability of drawing edges between agents. 
Simulations were run in Python 3 on a 2.2 GHz Intel Core i7 processor. 
In the simulations, we choose the message representation to be on $l = 32$ bits: 16 integer bits and 16 fractional bits, the statistical security size to be 80 bits and the Paillier moduli for each agent to have 2048 bits. 
With these chosen values, all 6 elements of the local contributions can be encoded into a single Paillier ciphertext in $\mr{pWSAh}^\ast$.

We present the simulation results for the solutions described in Sections~\ref{sec:private_weighted_sum_aggregation},~\ref{subsec:one_step} and~\ref{subsec:two_steps}, in both naive implementation (Section~\ref{subsubsec:naive}) and using packing~(Section~\ref{subsubsec:packed}). The running times are averaged over 50 instances and represent the total time it takes for an agent at a time step to generate and distribute the secret shares for the computation for itself and its neighbors \textit{and} to aggregate the contributions of its neighbors \textit{and} to compute and send out its own contribution to its neighbors. The shares are locally encrypted with an AES cipher with 128-bit key--for the offline centralized phase, each agent has an AES key with the dealer, and each pair of neighbors has their own AES key for the online decentralized phase. 
Figures~\ref{fig:centralized}, \ref{fig:decentralized}, \ref{fig:aggregated} and~\ref{fig:on_aggregated_u} show the times for an agent with the average connectivity degree, respectively the minimum and maximum connectivity degree (represented by the arrows).

\begin{figure}[!t]
 \centering
    \includegraphics[width=0.41\textwidth]{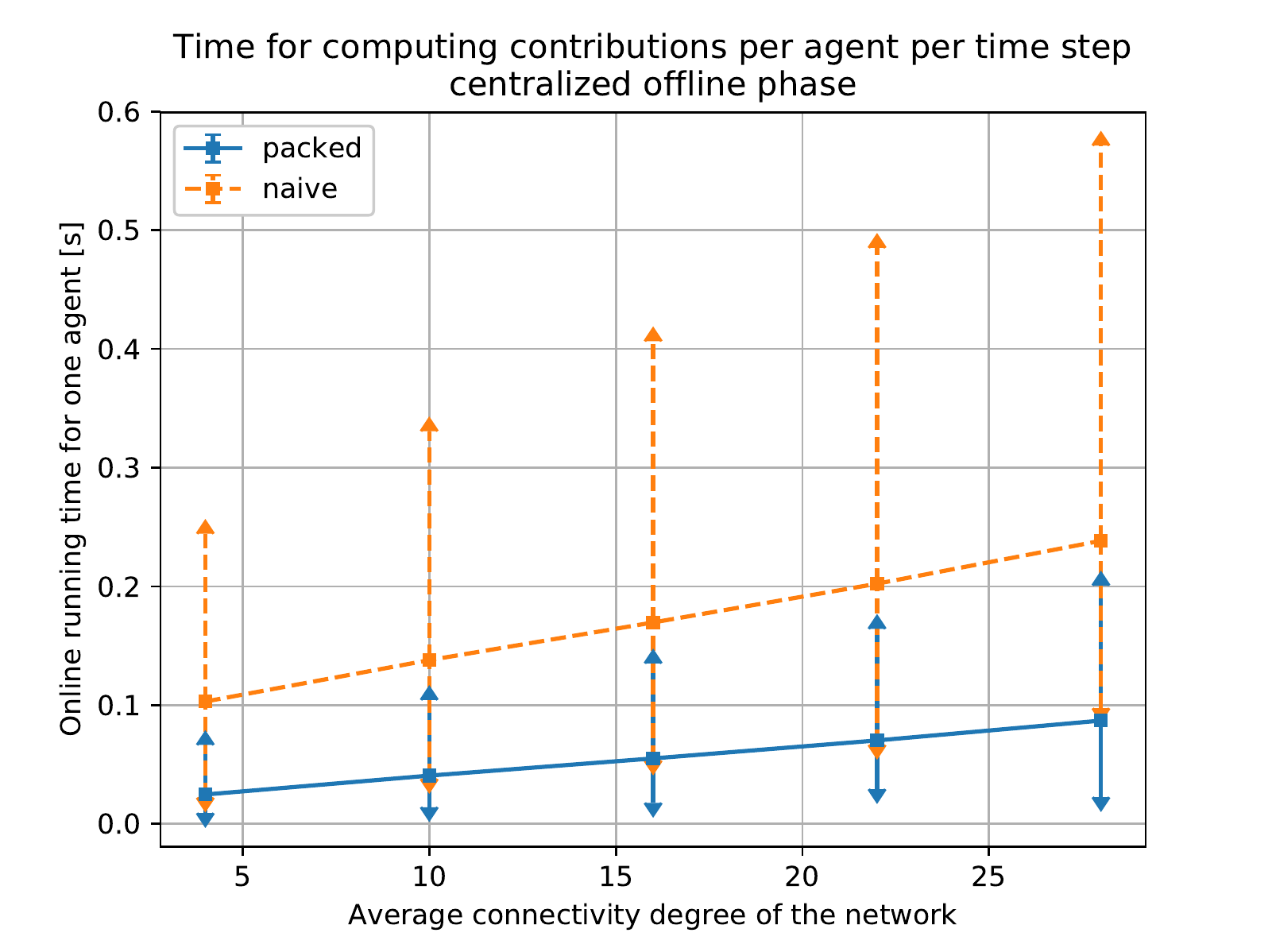}
  \caption{Average running times for the \tr{pWSAh}$^\ast$ scheme with the steps described in Section~\ref{sec:private_weighted_sum_aggregation} in a network of 50 agents.}
   \label{fig:centralized}
   ~
  \centering
    \includegraphics[width=0.41\textwidth]{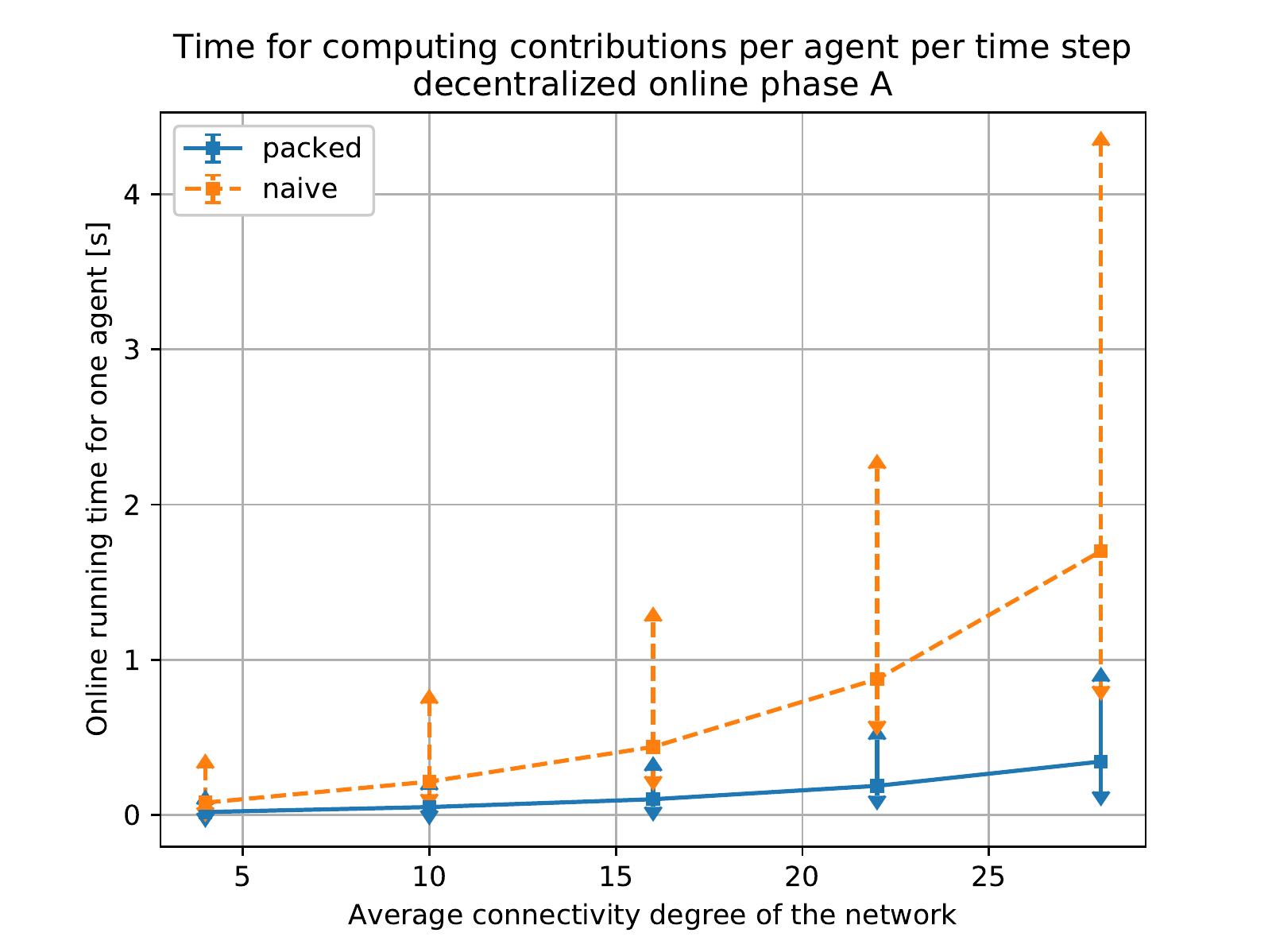}   
  \caption{Average running times for the \tr{pWSAh}$^\ast$ scheme with the steps described in Section~\ref{sec:distributed} in a network of 50 agents.}
  \label{fig:decentralized}
   ~
   \centering
    \includegraphics[width=0.41\textwidth]{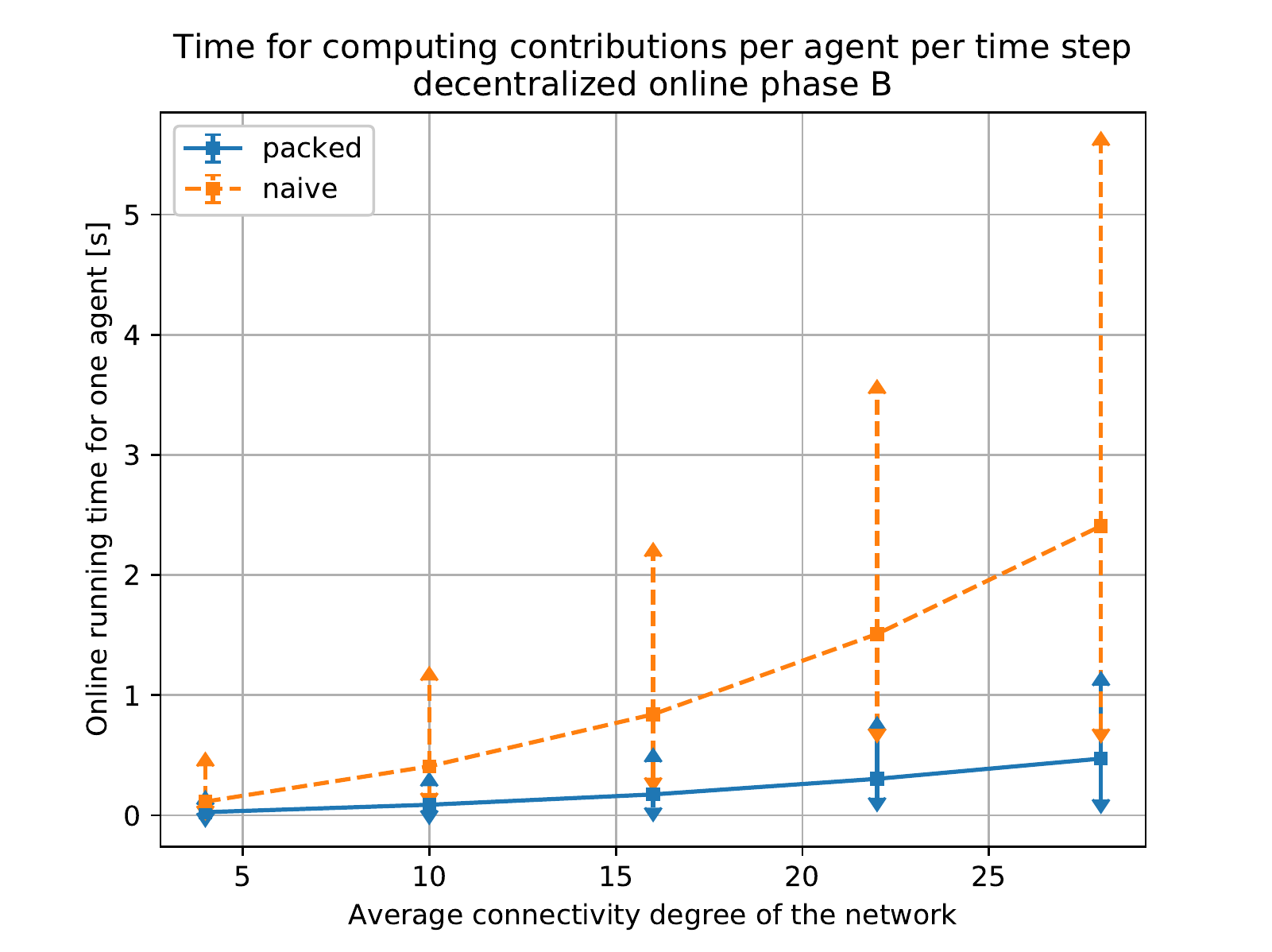}
  \caption{Average running times for the \tr{pWSAh}$^\ast$ scheme with the steps described in Section~\ref{subsec:two_steps} in a network of 50 agents. }
  \vspace{-12pt}
   \label{fig:aggregated}   
\end{figure}

Figure~\ref{fig:centralized} compares the running times for the local computation at each agent in a time step using the scheme~with~centralized offline generation of the secret shares, between~naive~and packed encryption. This method has~the~smallest online running time for agents, but the largest offline time~for the dealer that generates the shares for all agents, for many time steps (see Figure~\ref{fig:off_centralized}). 
Figure~\ref{fig:decentralized} compares the online running times between naive and packed encryption in the case of the one-step decentralized online generation of shares~and~Figure~\ref{fig:aggregated}~for the two-step decentralized online generation of shares. These methods have roughly an eight-fold increase in the online~time compared to the method that makes heavy use of a trusted dealer, but the dealer has less work to do in the offline phase. As expected, less security, in terms of reducing the collusion threshold, yields better online time (Figure~\ref{fig:decentralized} vs. Figure~\ref{fig:aggregated}).

In the centralized offline share generation scheme, packing decreases the maximum online time between 64\% and 71\%. In the online decentralized cases, where the agents are also responsible for generating, encrypting, sending and decrypting the shares, packing reduces the maximum online running time between 76\% and 80\%. Overall, we see that in the packed version, the sampling time needs to be at most 1.1 seconds.

Second, the communication load is reduced when using packing: one Paillier ciphertext of 0.256 KB is sent from the neighboring agents to the aggregating agent, instead of six ciphertexts amounting to 1.536 KB, and each agent sends four batches of 16 bytes AES-encrypted shares to each neighbor, instead of seven batches of 16 bytes.

A third substantial advantage of packing is decreasing the offline time, consisting of the weights encryption and the share generation and encryption at the system operator, depicted in Figure~\ref{fig:off_centralized}. Specifically, we see up to 80\% improvement in the offline running time when using packing.

\begin{figure}[!t]
 \centering
    \includegraphics[width=0.41\textwidth]{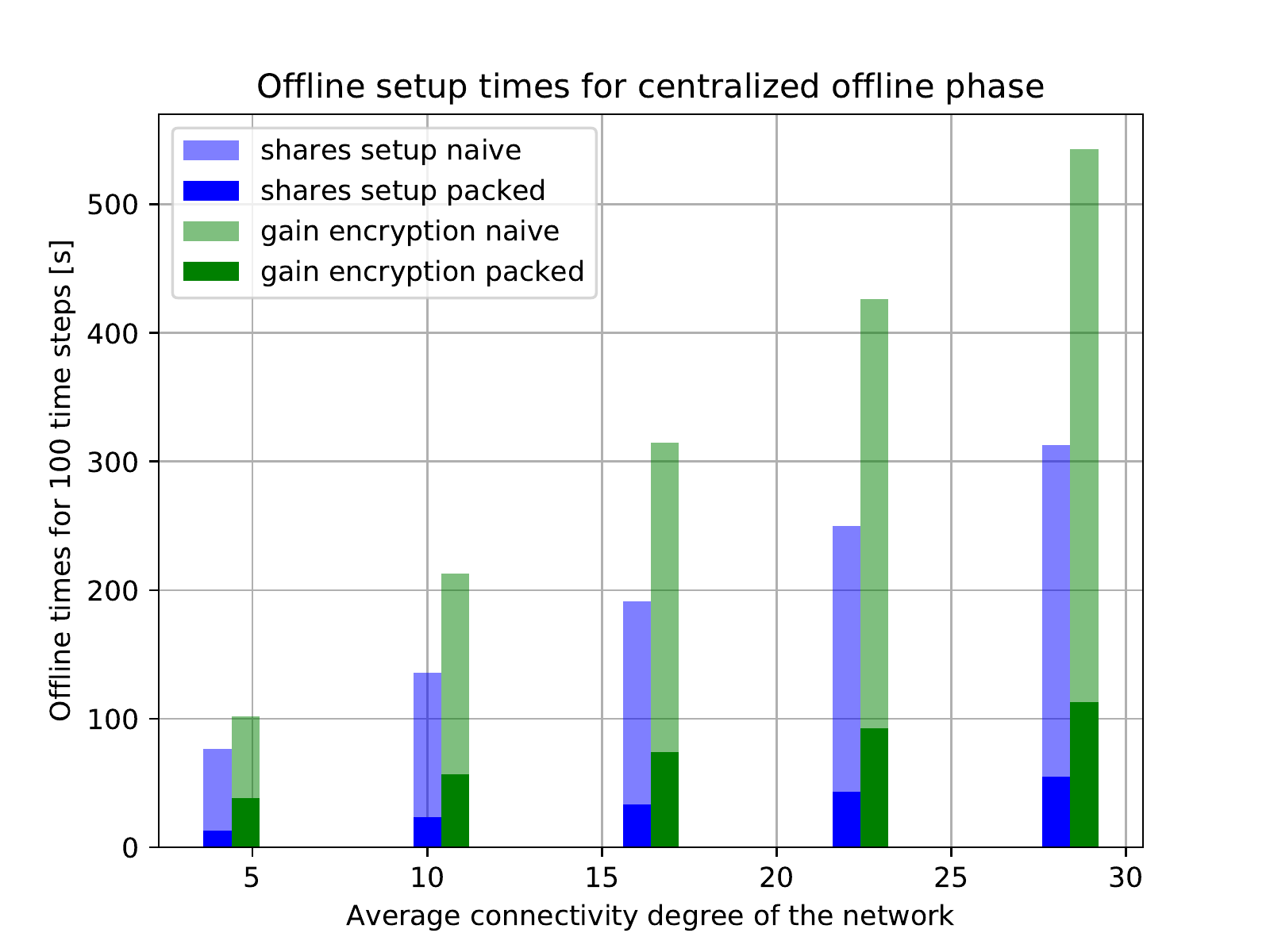}
  \caption{Offline setup phase running times for the \tr{pWSAh}$^\ast$ scheme as in Section~\ref{sec:private_weighted_sum_aggregation} in a network of 50 agents with shares 
  generated for 100 time steps. The offline setup times for the decentralized share generation schemes in Section~\ref{sec:distributed} consist only of the gain encryption times.}
   \label{fig:off_centralized}
   ~
 \centering
    \includegraphics[width=0.41\textwidth]{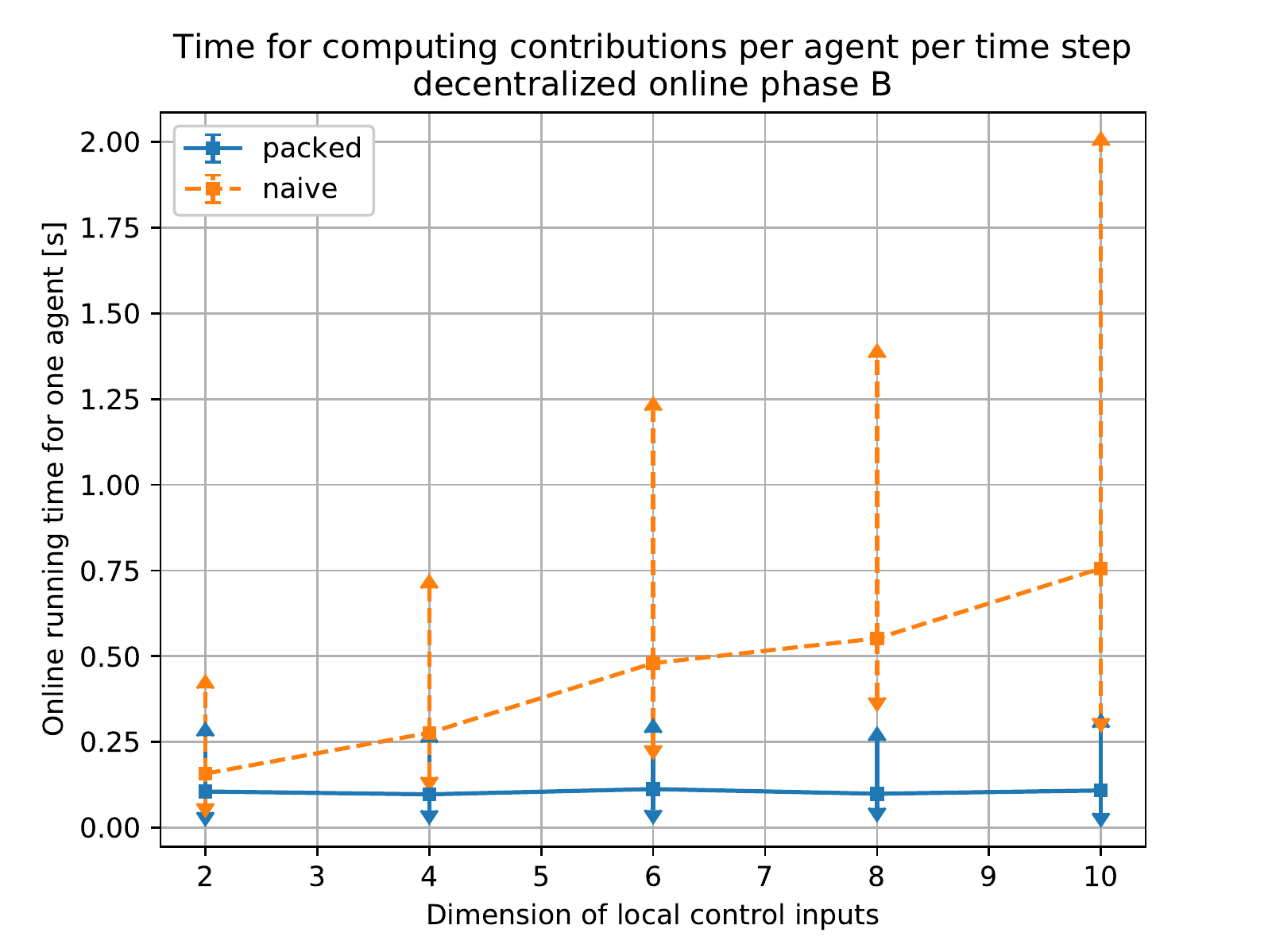}
  \caption{Average running times for the \tr{pWSAh}$^\ast$ scheme with the steps described in Section~\ref{subsec:two_steps} in a network of 25 agents.}
  \vspace{-8pt}
   \label{fig:on_aggregated_u}   
\end{figure}

To further illustrate the efficiency of packing, we show in Figure~\ref{fig:on_aggregated_u} how the performance improves with the number of control inputs, i.e., the number of values that are packed into one value, in the case of decentralized online share generation. We simulate a network of 25 agents with an average network connectivity degree of 10. Each agent has a local state of dimension 10 and local control input of dimension varying between 2 and 10, which are packed in one Paillier ciphertext. The online running time remains almost the same in the packed version, despite having more control inputs, compared to the increasing online time in the naive case.

%% file: Conclusion.tex
\section{Future work}\label{sec:conclusion}
In this work, we presented solutions for the problem of private weighted sum aggregation, where an aggregator has to obliviously obtain the sum of the weighted data of some agents. Depending on which participant has access to which piece of information, we can use different efficient solutions that exploit this information distribution. However, providing privacy in the honest-but-curious model might not be enough. We plan to investigate schemes that are private under more realistic security assumptions, namely for malicious adversaries, as well as under drop-out conditions. 

%% file: Appendix.tex
\appendix
\setcounter{theorem}{0}
     \renewcommand{\thetheorem}{\Alph{section}.\arabic{theorem}}
\setcounter{definition}{0}
     \renewcommand{\thedefinition}{\Alph{section}.\arabic{definition}}     
\subsection{Paillier's additively homomorphic encryption scheme}\label{app:AHE}
Consider the additive group of integers modulo $N$, $\ZN$, where $N=pq$ is a large modulus composed of two prime numbers of equal bit-length, $p$ and $q$, 
such that $\gcd(\phi(N),N)=1$. $\phi(N) = (p-1)(q-1)$ is the order of $(\ZN)^\ast$. 
Now consider the multiplicative group of integers modulo $N^2$, $\ZNm$. The order of $\ZNm$ is $N\phi(N)$. 
An important subgroup of $\ZNm$ is:
\begin{align}\label{eq:GammaN}
\begin{split}
	\Gamma_N :=& \{(1+N)^\alpha \bmod N^2 | \alpha \in \{0,\ldots,N-1\}\} \\
	=& \{1+\alpha N | \alpha \in \{0,\ldots,N-1\} \},
\end{split}
\end{align}
where the equality follows from the binomial theorem: \break$(1+N)^\alpha = 1 + \alpha N \bmod N^2 $. Computing discrete logarithms in $\Gamma_N$ is 
easy~\cite{Paillier99,Joye2013scalable}: given $x,y\in \Gamma_N$, we can find $\beta$ such that $y = x^\beta \bmod N^2$ by $\beta = (y - 1)/(x - 1) \bmod N$.

Another important subgroup in $\ZNm$ is:
\begin{align}
\begin{split}
	\mf{G}_N:=& \{x^N \bmod N^2 | x\in (\ZN)^\ast \}.
\end{split}
\end{align}
$\mf{G}_N$ has order $\phi(N)$. Computing discrete logarithms in $\mf{G}_N$ is as hard as computing discrete logarithms in $(\ZN)^\ast$ 
\cite{Joye2013scalable}.

We also have the modular equalities for $x\in\ZNm$:
\begin{equation}\label{eq:lagrange_thm}
	x^{\phi(N)} = 1 \bmod N, \quad x^{N\phi(N)} = 1 \bmod N^2 . 
\end{equation}

The Paillier scheme is defined using the previously described concepts. Specifically, the plaintext space is $\ZN$ and the ciphertext space is $\ZNm$. 
The public key is $(g,N)$, where $g$ is usually selected to be $1+N$, and the secret key $\big(\phi(N),(\phi(N))^{-1}\bmod N\big)$. The encryption is:
\[\mr{E}(x) = g^x r^N \bmod N^2,\quad r\in (\ZN)^\ast.\]
For a ciphertext $c\in\ZNm$, decryption uses~\eqref{eq:GammaN} 
and~\eqref{eq:lagrange_thm}:
\[\mr{D}(c) = (c^{\phi(N)} - 1)/N \cdot \phi(N)^{-1}\bmod N.\]

The Paillier scheme allows for homomorphic additions and multiplication by plaintexts, as follows:
\begin{align*}
	\mr{D}\big(\mr{E}(x)\cdot \mr{E}(y) \big) &= \mr{D}\big(\mr{E}(x+y)\big) = x+y \bmod N\\
	\mr{D}\big((\mr{E}(x))^y\big) &= \mr{D}\big(\mr{E}(xy)\big) = xy \bmod N.
\end{align*}

Under the Decisional Composite Residuosity assumption (i.e., distinguishing between an element from $\mf{G}_N$ and an element from $\ZNm$ is hard), the 
following holds:

\begin{theorem}\label{thm:Paillier}
	The Paillier cryptosystem is semantically secure~\cite{Paillier99}. \hfill$\diamond$
\end{theorem}

\subsection{Secret sharing}\label{app:SS}
Secret sharing is a tool that distributes a secret message to a number of parties, by splitting it into random shares. Specifically, \mbox{$t$-out-of-$n$} secret 
sharing splits a secret message into $n$ shares and distributes them to different parties; then, the secret message can be reconstructed by an authorized subset 
of parties, which have to combine at least $t$ shares. 

One common scheme is the additive \mbox{2-out-of-2} secret sharing scheme, which splits a secret message $m$ in a message space $\ZQ$ into two shares by: generating uniformly at random an element $s\in \ZQ$, adding it to the message and then distributing the shares $s$ and $m+s\bmod Q$. 
Both shares are needed in order to recover the secret. 
In some cases (see discussion in Section~\ref{subsec:multidim_scheme}), we prefer to sample $s\in (0, 2^{\log_2 Q+\lambda})$, for a statistical security parameter $\lambda$. 

\begin{theorem}\label{thm:SS}
Secret sharing is:\\ 
(a) perfectly secure when $s\in\ZQ$ ~\cite{Cramer15};\\
(b) $\lambda$-statistically secure when $s\in(0, 2^{\log_2Q+\lambda})$. \hfill$\diamond$
\end{theorem}

The proof of (b) follows from computing the advantage an adversary has for distinguishing between $m+s\in (0,2^{\log_2Q+1+\lambda})$ and a uniformly sampled random value $r\in(0,2^{\log_2Q+1+\lambda})$, which is $1/2+2^{-\lambda}$ (the statistical distance between $\ZQ$ and $(0,2^{\log_2Q+1+\lambda})$ is $2^{-\lambda}$).

\subsection{Aggregator obliviousness for \tr{pWSAh}}\label{app:privacy_def}
We give a formal description of the privacy requirements from Section~\ref{sec:problem_statement} as a cryptographic game between an adversary and a challenger, where the adversary $\mc A$ can corrupt agents and the aggregator. The weights $w_{i\in[M]}$ are constant over the time steps, so the adversary is forced to specify constant weights; in particular, the adversary will specify two sets of weights: $w_{i\in[M]}^{\mc A,0}$ and $w_{i\in[M]}^{\mc A,1}$. The security game $\mr{pWSAO}$ (private Weighted Sum Aggregator Obliviousness) is as follows:

\noindent\textbf{Setup.} The challenger runs the $\mr{Setup}$ algorithm and gives the public parameters $\mr{prm}$ to the adversary.

\noindent\textbf{Queries.} The adversary can submit compromise queries and encryption queries that are answered by the challenger. In the case of compromise 
queries, the adversary submits an index $i\in[M]$ to the challenger and receives $\mr{sk}_i$, which means the adversary corrupts agent~$i$. The set of 
the corrupted agents is denoted by~$\mc C$. In the case of encryption queries, the adversary is allowed one query per time step~$t$ and per agent~$i\in[M]$. 
The adversary submits $(i,t,w^{\mc A}_{i},x_i(t))$, where $w^{\mc A}_{i}=\{w^{\mc A,0}_{i},w^{\mc A,1}_{i}\}$, and the challenger first runs $\mr{sw}_i^{\mc A,0}=\mr{InitW}(\mr{prm},i,w^{\mc A,0}_i)$, $\mr{sw}_i^{\mc A,1} = \mr{InitW}(\mr{prm},i,w^{\mc A,1}_i)$ and returns $\mr{Enc}(\mr{prm},\mr{sw}_i^{\mc A,0},\mr{sk}_i,t,x_i(t))$ and $\mr{Enc}(\mr{prm},\mr{sw}_i^{\mc A,1},\mr{sk}_i,$ $t,x_i(t))$. 
The set of participants for which an encryption query was made by the adversary at time~$t$ is denoted by $\mc E(t)$.

\noindent\textbf{Challenge.} The adversary chooses a specific time step $t^\ast$. Let $\mc U^\ast$ denote the set of participants that were not compromised at 
the end of the game and for which no encryption query~was made at time $t^\ast$, i.e., $\mc U^\ast = ([M]\cup \{a\}) \setminus (\mc C \cup \mc E(t^\ast))$. The 
adversary specifies a subset of participants $\mc S^\ast \subseteq \mc U^\ast$. At this time $t^\ast$, for each agent~$i\in \mc S^\ast\setminus \{a\}$, the adversary 
chooses two plaintext series $x_i^0(t^\ast)$ and $x_i^1(t^\ast)$, along with $w^{\mc A,0}_{i}$ and $w^{\mc A,1}_{i}$, and sends them to the challenger. If $\mc S^\ast 
= \mc U^\ast$ and $a\notin \mc S^\ast$, i.e., the aggregator has been compromised, then, the values submitted by the adversary have to satisfy $\sum_{i\in 
\mc S^\ast} w^{\mc A,0}_{i}x_i^0(t^\ast) = \sum_{i\in\mc S^\ast} w^{\mc A,1}_{i}x_i^1(t^\ast)$. The challenger flips a random bit $b\in\{0,1\}$ and computes $\mr{Enc}(\mr{prm},\mr{InitW}(\mr{prm},i,w^{\mc A,b}_{i}),\mr{sk}_i,t,x_i^b(t^\ast))$, $\forall i\in \mc S^\ast$.
The challenger then returns the ciphertexts to the adversary. 

\noindent\textbf{Guess.} The adversary outputs a guess $b'\in\{0,1\}$ on whether $b$ is 0 or 1. The advantage of the adversary is defined as: 
\[\mr{\mb{Adv}}^{\mr{pWSAO}} (\mc A) := \left | \Pr [b' = b] - \frac{1}{2}\right|. \]

The adversary wins the game if it correctly guesses $b$. 

\begin{definition}\label{def:weighted_privacy}
A scheme $\mr{pWSAh} = (\mr{Setup},\mr{Enc},\mr{InitW},$ $\mr{AggrDec})$ achieves \textit{weighted sum aggregator obliviousness} if no probabilistic polynomial-time 
adversary has more than negligible advantage in winning this security game:
\begin{align*}
\mr{\mb{Adv}}^{\mr{pWSAO}} (\mc A) \leq \eta(\kappa).\tag*{$\diamond$}
\end{align*}
\end{definition}

\subsection{Proof of Theorem~\ref{thm:scheme_pWSA}}\label{app:proof}
\begin{proof}
We are going to treat two cases: I, the adversary does not corrupt the aggregator and II, the adversary corrupts the aggregator:
\[\Pr [b'=b] = \frac{1}{2} \Pr [b'=b| i\notin \mc C] + \frac{1}{2} \Pr [b'=b| i \in \mc C].\]

We will consider the stronger case where $\mc S^\ast = \mc U^\ast$; the weaker case where $\mc S^\ast \subseteq \mc U^\ast$ follows.

\noindent I. $a\notin \mc C$. From the compromise queries, the adversary holds the following information $\{\kappa, \mf{pk}, \{s_i(t)\}_{i\in \mc C}, 
\{w_{i}\}_{i\in\mc C}\}_{t\in [T]}$ and $\sum_{i\in \mc U} s_i(t) = -\sum_{i\in\mc C} s_i(t)$, for all $t\in[T]$. From the encryption queries at time~$t$, the adversary 
knows $\{c_i(t) = \mr{E}(w^{\mc A}_{i}x_i(t)) + s_i(t))\}_{i\in \mc E(t)}$. Then, the adversary chooses $t^\ast \in T$ and a series of $\{x_i^0(t^\ast)\}_{i\in\mc U^\ast}$ 
and $\{x_i^1(t^\ast)\}_{i\in\mc U^\ast}$ and receives from the challenger $\{c_i(t^\ast) = \mr{E}(w^{\ast,b}_{i}x_i^b(t^\ast)) + s_i(t^\ast))\}_{i\in \mc U^\ast}$. 

Because the adversary doesn't have the secret key of the Paillier scheme and does not have the individual secrets of the uncorrupted agents, the following holds, where $\eta_1(\kappa), \eta_2(\kappa)$ are negligible functions, according to Theorems~\ref{thm:Paillier} and~\ref{thm:SS}:
\begin{align}\label{eq:inotin1}
\begin{split}
	&\Pr [\mc A \text{ breaks Paillier scheme}] \leq \eta_1(\kappa),\\
	&\Pr[\mc A \text{ breaks secret sharing}] \leq \eta_2(\kappa), \\
	&\Pr[b'=b | i\notin \mc C] \leq \frac{1}{2} + \eta_1(\kappa) \eta_2(\kappa).
\end{split}
\end{align}

\noindent II. $a\in \mc C$. From the compromise queries, the \mbox{adversary} holds the following information $\forall t\in[T]$: $\{\kappa, \mf{pk}, \{s_i(t)\}_{j\in \mc C},$ 
$\{w_{i}\}_{i\in\mc C}, \mf{sk}\}_{t\in [T]}$, and $\sum_{i\in \mc U} s_i(t) = -\sum_{i\in\mc C} s_i(t)$. From the encryption queries, and after using 
$\mf{sk}$ to decrypt, the adversary knows $\{p_i(t) = w^{\mc A}_{i}x(t) + s_i(t)\}_{i\in\mc E(t)}$. Then, the adversary chooses $t^\ast \in T$ and a series of 
$\{x_i^0(t^\ast)\}_{i\in\mc U^\ast}$ and $\{x_i^1(t^\ast)\}_{i\in\mc U^\ast}$, such that $\sum_{i\in\mc U^\ast} w^{\mc A,0}_{i}x_i^0(t^\ast) = \sum_{i\in \mc U^\ast} w^{\mc A,1}_{i}x_i^1(t^\ast)$ and receives from the challenger $\{c_i(t^\ast) = \mr{E}(w^{\mc A,b}_{i}x_i^b(t^\ast)) + s_i(t^\ast) \}_{i\in \mc U^\ast}$. The adversary uses the secret key of the Paillier scheme to decrypt the individual ciphertexts and obtains $p_i(t^\ast) = w^{\mc A,b}_{i}x_i^b(t^\ast) + s_i(t^\ast)\mod N$, for $i\in\mc U^\ast$. Because the secret shares of zero are different for each time $t\neq t^\ast$, the adversary cannot infer information about the challenge query from the previous encryption queries.

Then, the probability that the adversary wins is the probability that the adversary breaks secret sharing:
\begin{align}\label{eq:iin1}
\begin{split}
	&\Pr[\mc A \text{ breaks secret sharing}] \leq \eta_2(\kappa),\\
	&\Pr[b'=b | i\in \mc C] \leq \frac{1}{2} + \eta_2(\kappa).
\end{split}
\end{align}

From~\eqref{eq:inotin1} and~\eqref{eq:iin1}: 
$\mb{Adv}^{\mr{pWSAh}} (\mc A) \leq \eta_2(\kappa)$.
\end{proof}